\documentclass[lettersize,journal]{IEEEtran}
\usepackage{amsmath,amsfonts,amssymb,bm,amsthm}
\usepackage{color}
\usepackage{algorithmic}
\usepackage{array}
\usepackage{caption}
\usepackage[justification=raggedright,singlelinecheck=false,font=footnotesize]{caption}
\usepackage[subrefformat=parens,labelformat=simple,justification=centerlast]{subcaption}
\usepackage{textcomp}
\usepackage{algorithm, algorithmic}
\usepackage{stfloats}
\usepackage{url}
\usepackage{verbatim}
\usepackage{graphicx}
\def\BibTeX{{\rm B\kern-.05em{\sc i\kern-.025em b}\kern-.08em
    T\kern-.1667em\lower.7ex\hbox{E}\kern-.125emX}}
\usepackage{balance}
\usepackage{threeparttable}
\usepackage{pifont}
\usepackage{bbding}
\usepackage{multirow}
\usepackage{makecell}
\usepackage{booktabs}
\usepackage{epstopdf}
\usepackage{epsfig, latexsym, times}
\usepackage{float}
\usepackage{subcaption}
\usepackage{cases}
\usepackage{stfloats}
\usepackage[colorlinks=true, allcolors=blue]{hyperref}

\newtheorem{proposition}{Proposition}

\makeatletter
\renewenvironment{proof}[1][\proofname]{\par
	\pushQED{\qed}%
	\normalfont \topsep6\p@\@plus6\p@\relax
	\trivlist
	\item[\hskip\labelsep
	\itshape 
	#1\@addpunct{:}]\ignorespaces  
}{
	\popQED\endtrivlist\@endpefalse
}
\makeatother

\renewcommand{\proofname}{Proof}

\begin{document}
\title{\color{black}Hybrid Beamforming for RIS-Assisted Multiuser Fluid Antenna Systems}
\author{Jiangong Chen,
	    Yue Xiao, \IEEEmembership{Member,~IEEE},
        	    Zhendong Peng,
	   Jing Zhu,
	   Xia Lei,\\
	   Christos Masouros, \IEEEmembership{Fellow,~IEEE},
	   and Kai-Kit Wong, \IEEEmembership{Fellow,~IEEE}
\vspace{-10mm}

\thanks{
A portion of this paper was presented at the 2024 IEEE International Conferenceon Communications Workshops, Denver, CO, USA, June 2024 (J. Chen, Y. Xiao, J. Zhu, Z. Peng, X. Lei and P. Xiao, “Low-Complexity Beamforming Design for RIS-Assisted Fluid Antenna Systems,” in Proc. 2024 IEEE International Conferenceon Communications Workshops (ICC Workshops), Denver, CO, USA,2024, pp. 1377-1382, doi: 10.1109/ICCWorkshops59551.2024.10615874).

J. Chen, Y. Xiao, X. Lei are with the National Key Laboratory of Science and Technology on Communications, University of Electronic Science and Technology of China, Chengdu, China (e-mail: jg\_chen@std.uestc.edu.cn, \{xiaoyue, leixia\}@uestc.edu.cn).

J. Zhu is with 5G and 6G Innovation centre, Institute for Communication Systems (ICS) of University of Surrey, Guildford, GU2 7XH, UK (e-mail: j.zhu@surrey.ac.uk).
	
Z. Peng is with the Department of Electrical and Computer Engineering, The University of British Columbia, Vancouver, BC V6T 1Z4, Canada (e-mail: zhendongpeng@ece.ubc.ca).
	
C. Masouros and K. K. Wong are with the Department of Electronic and Electrical Engineering, University College London, London, UK (e-mail: \{c.masouros, kai-kit.wong\}@ucl.ac.uk). K. K. Wong is also with the Department of Electronic Engineering, Kyung Hee University, Yongin-si, Gyeonggi-do 17104, Korea.

}
}
\markboth{Journal of \LaTeX\ Class Files,~Vol.~18, No.~9, September~2020}%
{How to Use the IEEEtran \LaTeX \ Templates}

\maketitle

\begin{abstract}
Recent advances in reconfigurable antennas have led to the new concept of the fluid antenna system (FAS) for shape and position flexibility, as another degree of freedom for wireless communication enhancement. This paper explores the integration of a transmit FAS array for hybrid beamforming (HBF) into {\color{black}a reconfigurable intelligent surface (RIS)-assisted} communication architecture for multiuser communications in the downlink, corresponding to the downlink RIS-assisted multiuser multiple-input single-output (MISO) FAS model (Tx RIS-assisted-MISO-FAS). By considering Rician channel fading, we formulate a sum-rate maximization optimization problem to {\color{black} alternately} optimize the HBF matrix, the RIS phase-shift matrix, and the FAS position. {\color{black}Due to the strong coupling of multiple optimization variables, the multi-fractional summation in the sum-rate expression, the modulus-1 limitation of analog phase shifters and RIS, and the antenna position variables appearing in the exponent, this problem is highly non-convex}, which is addressed through the block coordinate descent (BCD) framework in conjunction with semidefinite relaxation (SDR) and majorization-minimization (MM) methods. To reduce the computational complexity, we then propose a low-complexity grating-lobe (GL)-based telescopic-FAS (TFA) {\color{black}with multiple delicately deployed RISs} under the sub-connected HBF architecture and the line-of-sight (LoS)-dominant channel condition, to allow closed-form solutions for the HBF and TFA position. Our simulation results illustrate that the former optimization scheme significantly enhances the achievable rate of the proposed system, while the GL-based TFA scheme also provides a considerable gain over conventional fixed-position antenna (FPA) systems, requiring statistical channel state information (CSI) only and with low computational complexity.
\end{abstract}


\vspace{-2mm}
\section{Introduction}
\IEEEPARstart{T}{he fast-growing} demands in wireless communications continue to pile pressure on existing technologies when heading to the sixth-generation (6G) mobile communication \cite{6G,Tariq-6G}. To prepare for the next generation, recent research has spurred the engineering of electromagnetic environments, leading to the vision of smart radio environments (SRE) \cite{RISSRE,Wong-sre2021} or smart propagation engineering (SPE) \cite{ISACSRE}.

As a promising SPE/SRE technology, reconfigurable intelligent surface (RIS) has garnered significant attention for its ability to control signal propagation properties by intelligently adjusting its reflecting units \cite{RISSurvey}. The additional diversity gain and spatial degrees of freedom by RIS have sparked extensive studies across many applications, such as channel estimation \cite{CE1}, precoding or beamforming \cite{RISBF2,RISBF3,RISBF7,RISBF10,RISBF5,RISBF8,RISBF9,RISBF6}, physical layer security (PLS) \cite{PLS2, PLS3}, interference cancellation \cite{INCan}, and sensing \cite{Loc1,Loc3,RISISAC,RISISAC2}.

Clearly, joint beamforming between the transmitter (e.g., the base station (BS)) and the RIS is of great importance for its promising performance. In order to optimize the intertwined variables of active beamforming at the transmitter and passive beamforming at the RIS, as well as the modulus-1 constraints of the RIS phase shift coefficients, several techniques have been widely adopted, such as the alternating optimization (AO) \cite{RISBF2,RISBF3}, alternating direction method of multipliers (ADMM) \cite{RISBF7}, block coordinate descent (BCD) \cite{RISBF10}, fractional programming (FP) \cite{RISBF2,RISBF5}, semidefinite relaxation (SDR) \cite{RISBF2,RISBF3}, majorization-minimization (MM) \cite{RISBF8,RISBF9}, and projected gradient method (PGM) \cite{RISBF6}. Robust beamforming with consideration of channel state information (CSI) errors was also addressed in \cite{RISRobustBF} by adopting the S-procedure and the Bernstein-type inequality. But the limitations of RIS lie in the complexity of getting the CSI and optimization, without making it practically unviable.

Meanwhile, a new form of reconfigurable antenna technologies, referred to as fluid antenna (FA), has emerged to advocate shape and position flexibility in antenna and FA system (FAS) has recently become a new degree of freedom to enhance the performance of wireless communications \cite{Wong-ell2020,FAS_Lee}. FA is also known as the movable antenna with the latter term specifying the implementation of position reconfigurability by mechanically movable antennas \cite{zheng2023flexibleposition,Zhu-Wong-2024}. In terms of implementation, other than mechanical movement, FAS can be best realized by meta-materials \cite{liquid_antenna8,Hoang-2021,Deng-2023} or reconfigurable pixels \cite{pixel_antenna,Jing-2022}. Proof-of-concepts for FAS were reported in \cite{Shen-tap_submit2024,Zhang-pFAS2024}.

Unlike RIS which modifies the propagation channel through smart reflection to produce desirable channel conditions, FAS utilizes position\footnote{Note that FAS has shape flexibility of antenna as well that includes antenna orientation and polarization. Depending on the channel model, this can be an effective degree of freedom (DoF) in designing communication systems.} diversity of antennas at the transmitter and/or receiver, thereby affecting the array's steering vector or the correlation function between activated positions (a.k.a.~ports). Such degrees of freedom can be utilized to avoid deep fade, enhance channel gain, reduce potential interference, improve sensing and communication trade-offs, etc.  

FAS was first introduced by Wong {\em et al.} in 2020 \cite{FA1,FAS} where a transmitter with a fixed-position antenna communicating to a receiver with a FA, referred to as the Rx-single-input single-output (SISO)-FAS model according to the nomenclature in \cite{New-submit2024}, in rich scattering channels was studied. These works investigated the ergodic rate and outage probability. Later in \cite{FAChan,ramirez2024new}, efforts have been made to improve the channel model to characterize more accurately the spatial correlation among different ports of FAS. Most recently, the diversity order of Rx-SISO-FAS has been studied in \cite{10130117} while the case where multiple FAs are used at both ends referred to as the dual-multiple-inupt multiple-output (MIMO)-FAS model, was also addressed in \cite{10303274}, with its diversity and multiplexing trade-off analyzed. Evidently, the promising performance of FAS relies on the decision of selecting the optimal port which has led to efforts in \cite{FAPortSel} using machine learning methods to approximate the port selection process. FA port optimization has also been applied together with multiuser beamforming \cite{FAPortSelMUBF}, simultaneous wireless information and power transfer (SWIPT) networks \cite{FAPortSelSWIPT}, and relay-aided networks \cite{FAPortSelRelay}. Also, inspired by the deep fading effect of multiuser interference, FAS has also been proven to be effective for multiple access, without relying on precoding \cite{FA2,10066316,FASurvey}. {\color{black}Recently, FA has been proven to have the potential for striking better tradeoff performance in integrated sensing and communication (ISAC) systems \cite{FAISAC}. In order to obtain the CSI of the uplink multiuser FAS systems, a low-complexity and high-precision channel estimation method was proposed in \cite{xu2023channel}.}

Another line of work on FAS focuses on finite-scattering channels and emphasizes changes in the steering vectors, a.k.a. array responses, of the transmit and/or receive arrays within the spatial channel model. Several results along this line came under the name of movable antennas \cite{MASurvey}. In \cite{MAAnalogBF1,MAAnalogBF2}, significant advantages in multi-direction beamforming and interference nulling under single-path line-of-sight (LoS) channel conditions were demonstrated. A field-response-based channel model was subsequently proposed in \cite{MAChan1,MAChan2} for a multi-path channel model. An uplink system with FAS-aided users transmitting to a BS with multiple fixed-position antennas, i.e., the uplink Tx-SIMO-FAS model \cite{New-submit2024}, was considered, and antenna position and power allocation strategies were designed by using zero-forcing and minimum mean square error (MMSE) techniques \cite{MAOptMU}. While \cite{MAOptMU} considered the power minimization problem with a rate constraint, \cite{xu2023capacity} studied the capacity maximization of the same model. The downlink counterpart which is referred to as the dual-MIMO-FAS model \cite{New-submit2024}, was later addressed in \cite{MAOptMUBF1,MAOptMUBF2}. An uplink non-orthogonal multiple access (NOMA) system has also been considered in conjunction with FAS in \cite{MAOptNOMA}, trying to jointly optimize the antenna positions (hence the steering vectors), the decoding order and the power control for rate maximization. In \cite{MAOptMei,MAOptMeiPLS}, a discrete port selection problem was formulated and solved by graph theory. {\color{black}Similarly, discrete port selection and beamforming optimization were considered in \cite{MAOptYifei}, which is solved by generalized Bender’s decomposition for a globally optimal solution.} The work in \cite{6DMA1} and \cite{6DMA2} further found rotation useful to avoid unwanted reflections from its other movable arrays at the same BS when movable arrays are used. Optimizing the steering vectors via FAS has also found applications in ISAC \cite{ISAC_FAS}, increasing the dimensionality of MMSE \cite{Wong-vfas}, over-the-air computation \cite{FAPortSelOAC} etc. Note that finite scattering channels appear to simplify the CSI estimation process with more sparsity, which was addressed in \cite{10375559}.

Although RIS is a much more developed area than FAS, it is understandable that the intersection between RIS and FAS is quite an open area; yet their synergy is believed to be key to the vision of SPE/SRE \cite{ISACSRE,RISFASurvey}. There have been few attempts to combine RIS and FAS in a productive manner. In \cite{FAPortSelRIS}, it was shown that FAS could simplify RIS processing and make RIS effective when only statistical CSI was available. More recently, \cite{RISFA2} obtained the analytical expressions for the outage probability and delay outage rate for the Tx-SISO-FAS model with a broken direct link and the help of RIS. The results in \cite{RISFA2} illustrated the extraordinary synergy that exists between RIS and FAS. Despite this, much has yet to be explored.

{\color{black}In this paper, we extend the single-user model in \cite{FAPortSelRIS} to a multiuser downlink model, and our originality lies in the introduction of FAS and hybrid beamforming (HBF) architecture into the BS for enhanced capability and reduced overhead. We refer to this model as RIS-assisted multiuser multiple-input single-output (MU-MISO) FAS.}\footnote{According to \cite{New-submit2024}, this model is referred to as downlink (multiuser) Tx-MISO-FAS where a multi-RF chains FAS is used at the BS side to support multiple users with fixed-position antennas in the downlink.} The main contributions of this paper are summarized as follows.
\begin{itemize}
 \item {\color{black}This is the first study considering RIS for HBF in the MU-MISO FAS.} With perfect CSI at the BS transmitter side, a sum-rate maximization problem exploiting position flexibility, and joint active and passive beamforming is formulated. {\color{black}Different from existing works, we consider fully digital, fully-connected, and sub-connected transmitter architectures. Furthermore, we extend the system model and problem formulation to the case with an arbitrary number of RIS.} {\color{black}Due to the strong coupling of multiple optimization variables, the multi-fractional summation in the sum-rate expression, the modulus-1 limitation of analog phase shifters and RIS, as well as the antenna position variables appearing in the exponent, this problem is highly non-convex.} This problem is decoupled by the BCD \cite{BCD} and FP \cite{FP} frameworks. The three decoupled nonconvex optimization subproblems are further relaxed into convex ones through the SDR \cite{SDR} and MM methods \cite{MM}, and solved iteratively.
 
\item To lessen the CSI requirement and optimization complexity, a relatively trivial position switching and its hardware implementation, referred to as telescopic FA (TFA), is proposed, which needs only statistical CSI, i.e., the LoS component, to operate. {\color{black}For single-user and single-RIS case, TFA can direct the main lobe (ML) and grating lobe (GL) towards the user equipment (UE) and RIS, respectively, with small overhead for position switching and optimization \cite{FAPortSelRIS}. However, to avoid excessive interference in space, the UE and the RIS must maintain a certain angle separation.
 }

\item {\color{black}Given the UE coverage limitation in a single-RIS case scenario, multiple RISs with delicate placements are further employed in this paper to serve UE in any direction. First, we summarize the pattern of GL formation in a uniform linear array (ULA). Based on this, we group UEs and RISs based on their geolocations, with each group served by a single TFA sub-array.  Using the GL effect brought by the TFA, a closed-form suboptimal solution for both analog beamforming (ABF) and the TFA positions can be directly derived. Subsequently, passive beamforming of the RIS is achieved using the generalized Rayleigh-Ritz theorem. Finally, digital beamforming (DBF) is done by MMSE precoding to eliminate inter-user interference. }

\end{itemize}

{\color{black}From our simulation results, it is evident that FA positions, HBF at the transmitter side, and passive beamforming at the RISs are all crucial for system performance. Using the optimization framework proposed in this paper, FA can even help the system achieve a performance leap across transmitter architectures. Additionally, due to the design philosophy centred on local optimization, limited antenna mobility pattern, partial knowledge of CSI, and rudimentary interference management, the proposed TFA underperform relative to the FA based on optimization algorithms. However, simulation results verify that TFA also offers performance gains compared to FPA. Moreover, we believe that TFA substantially reduces the costs associated with antenna movement, algorithmic complexity, and channel estimation challenges, making it a more practical solution for real-world applications. }

The rest of this paper is organized as follows. In Section \ref{sec:model}, we present the system model and problem formulation for the {\color{black}RIS-assisted} MU-MISO FAS. Section \ref{sec:opt} discusses the optimization framework while Section \ref{sec:tfa} proposes the {\color{black}low-complexity TFA design} under the LoS-dominant condition and sub-connected HBF architectures. Finally, Section \ref{sec:results} presents the simulation results and Section \ref{sec:conclude} concludes this paper.

\emph{Notations}: In this paper, we use lower-case letters, lower-case bold letters, and capital bold letters to denote scalars, vectors and matrices respectively. The operators for transpose, conjugate, conjugate transpose and inverse are denoted by  $\left(\cdot\right)^{\text{T}}$, $\left(\cdot\right)^{\text{H}}$, $\left(\cdot\right)^{\text{*}}$ and $\left(\cdot\right)^{-1}$ respectively. $\text{Tr}\left(\mathbf{A}\right)$ and $\text{Rank}\left(\mathbf{A}\right)$ stand for the trace and rank of matrix $\mathbf{A}$. The $i$th element of vector $\mathbf{a}$ is $[\mathbf{a}]_i$ and the $(i,j)$th element of matrix $\mathbf{A}$ is $[\mathbf{A}]_{i,j}$. $|a|$, $\|\mathbf{a} \|_2$, $\|\mathbf{A} \|_2$, and $\|\mathbf{A} \|_F$ respectively represent the modulus of scalar $a$, $\ell$-2 norm of vector $\mathbf{a}$, induced 2-norm of matrix $\mathbf{A}$ and Frobenius norm of matrix $\mathbf{A}$. The operations for extracting the imaginary part, real part, and phase of a complex variable are denoted by $\Im\left(\cdot\right)$, $\Re\left(\cdot\right)$ and $\angle\left(\cdot\right)$ respectively. $\otimes$ represents the Kronecker product. Moreover, $\text{diag}\left(\cdot\right)$ and $\text{blkdiag}\left(\cdot\right)$ denote the diagonal and block-diagonal matrix. Also, $\jmath = \sqrt{-1}$ is the imaginary unit. $\mathbf{I}_N$ denotes the $N$-dimensional unit array. {\color{black}Lists with the most important variables are provided in Table \ref{Variables}.}
\renewcommand{\arraystretch}{1.5}
\begin{table}[htbp]
	\color{black}
	\caption{List of variables.}
	\centering
	\begin{threeparttable}
		\begin{tabular}{*{2}{>{\centering\arraybackslash}m{2.5cm}>{\centering\arraybackslash}m{5.5cm}}} 
			\Xhline{2pt}
			\textbf{Parameter} & \textbf{Description}  \\
			\Xhline{1pt}
			${\mathbf{a}}_N^{}$, ${\mathbf{\dot a}}_N^{}$, ${\mathbf{\dot b}}_N^{}$ &  Steering vectors of  $N$-element FPA, FA, TFA\footnote{}\\		
			${{\bf{h}}}^{\rm b,u}_{ k}$, $\mathbf{H}^{\rm b,r}_{l}$, $\mathbf{h}^{\rm r,u}_{l,k}$& Channel coefficients\footnote{}\\
			$\mathbf{W}$, $\mathbf{V}$, $\mathbf{F}$ & Beamformer (DBF, ABF, Overall) \\
			$\mathbf{E}_l$  &  Phase shift matrix of the $l$th RIS \\
			$\mathbf{z}$ & Antenna position vector of the FA array\\
			$\mathbf{g}_k$ & Overall channel coefficients to the $k$th UE\\
			${{\mathbf{\Omega}}_k}$, ${{\mathbf{\Upsilon }}}$, ${{\mathbf{\Gamma }}_k}$, $\mathbf{\Xi }$  &  Beamforming covariance matrix (DBF, ABF, Overall beamforming, RIS)\\
			\Xhline{2pt}
		\end{tabular}%
	\end{threeparttable}
		\begin{tablenotes}
		\footnotesize
		\item $^3$Dotted variables are introduced for FA systems to distinguish the ones from FPA systems.
		\item $^4$In this paper, subscripts are used to denote indices while superscripts indicate channel types. For instance, the variable $\mathbf{h}^{\rm r,u}_{l,k}$ represents the channel coefficient from the $l$th RIS to the $k$th UE.
	\end{tablenotes}
	\label{Variables}%
\end{table}

\begin{figure}
\centering
\includegraphics[scale=0.8]{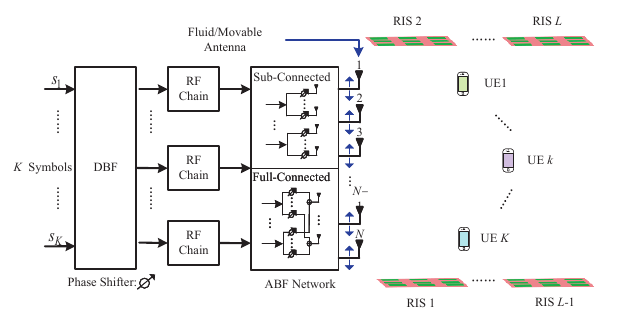}
\caption{\color{black}System model of the RIS-assisted MU-MISO FAS.}\label{SystemModel}
\vspace{-3mm}
\end{figure}

\section{System Model}\label{sec:model}
As shown in Fig.~\ref{SystemModel}, {\color{black}we consider a narrowband downlink RIS-assisted MU-MIMO FAS}, consisting of a BS, {\color{black}$L$ } $M$-unit uniform planar array (UPA) RISs, and $K$ UEs, each with a single antenna. The BS is equipped with a $K$-RF chain, $N$-port linear FA array to serve the UEs.\footnote{We assume that the numbers of UE and RF chains are identical.}

\subsection{Channel Model}
The antenna position vector (APV) of the FA array at the BS is denoted by ${\mathbf{z}} \triangleq {[z_1,\; \ldots ,\;z_N^{}]^{\text{T}}}$. We denote the elevation angles of BS-UE, the elevation-azimuth angle pairs of BS-RIS, and the elevation-azimuth angle pairs of RIS-UE as $\theta^{\rm b,u}_{k}$, $\{\theta^{\rm b,r}_l,\phi^{\rm b,r}_l\}$, and $\{\theta^{\rm r,u}_{k,l},\phi^{\rm r,u}_{k,l}\}$, respectively. Here, $k \in \mathcal{K} \triangleq \{1,\dots,K\} $ is the UE index, {\color{black}$l\in \mathcal{L} \triangleq \{1,\;\ldots,\;L\}$} is the RIS index, while the superscripts indicate the channel type. As such, the steering vector of the transmit FA array can be written as
\begin{equation}\label{FASV}
{\mathbf{\dot a}}_N^{}\left( {\theta ,{\mathbf{z}}} \right) \triangleq {\left[ {1,\;{e^{ - \jmath 2\pi \dot f_1^{}(\theta )}},\; \dots ,\;{e^{ - \jmath 2\pi \dot f_{N - 1}^{}(\theta )}}} \right]^{\text{T}}},
\end{equation}
where ${{\dot f}_n} (\theta) \triangleq  {z_n}  {\rm{cos}}(\theta )/\lambda$, for $n= 1,\dots,N$, represents the spatial frequency of the transmit FA array, with $\lambda$ denoting the wavelength and $z_n$ being the $n$th element of $\mathbf{z}$. On the other hand, the steering vector for the RIS can be found as
\begin{equation}\label{RISSV}
{{\bf{a}}_M}\left( \theta,\phi  \right) \triangleq {{\bf{a}}_{{M_1}}}\left( \theta, \phi  \right) \otimes {{\bf{a}}_{{M_2}}}\left( {\theta ,\phi } \right),
\end{equation}
where $M = M_1 \times M_2 $, and
\begin{align}
{{\bf{a}}_{{M_1}}}\left( {\theta ,\phi } \right)& \triangleq \left[ {1,\;{e^{ - \jmath 2\pi f_1^{y}(\theta,\phi) }}, \dots } \right.{\left. {,\;{e^{ - \jmath 2\pi f_{{M_1} - 1}^{y} (\theta,\phi)}}} \right]^{\rm{T}}},\\
{{\bf{a}}_{{M_2}}}\left( {\theta ,\phi } \right) &\triangleq \left[ {1,\;{e^{ - \jmath 2\pi f_1^{x} (\theta,\phi)}}, \dots } \right.{\left. {,\;{e^{ - \jmath 2\pi f_{{M_2} - 1}^{x} (\theta,\phi)}}} \right]^{\rm{T}}},
\end{align}
in which we have $f_{m_1}^{y} (\theta,\phi) \triangleq \left( {m_1 - 1} \right)d_{\rm RIS}\sin (\theta )\sin \left( \phi  \right)/\lambda$ and $f_{m_2}^{x} (\theta,\phi) \triangleq \left( {m_2 - 1} \right)d_{\rm RIS}\sin (\theta )\cos \left( \phi  \right)/\lambda $ representing the corresponding spatial frequencies for RIS, where $d_{\rm RIS} = \lambda/2$ denotes the element spacing of RIS.

With the assumption of a block-fading and LoS-dominant Rician channel model \cite{LoSChan} among the BS, RIS, and UE, the direct link channel vector between the transmit linear FA array and the $k$th UE is modeled as
\begin{equation}
{{\bf{h}}}^{\rm b,u}_{ k}  \triangleq {\beta^{\rm b,u}_{ k}}\left( {\sqrt {\frac{{{\kappa^{\rm b,u}_{ k}}}}{{{\kappa^{\rm b,u}_{ k}} + 1}}} {{{\bf{\bar h}}^{\rm b,u}_{ k}}} + \sqrt {\frac{1}{{{\kappa^{\rm b,u}_{ k}} + 1}}} {{{\bf{\tilde h}}^{\rm b,u}_{ k}}}} \right),
\label{H1}
\end{equation}
where ${{{\bf{\bar h}}^{\rm b,u}_{k}}} = {{{\bf{\dot a}}}_N}\left( {\theta^{\rm b,u}_{k} ,{\bf{z}}} \right)$ is the LoS component, ${{{\bf{\tilde h}}}^{\rm b,u}_k}$ is the NLoS component whose elements are independently and identically distributed (i.i.d.) complex Gaussian distributions with zero mean and unit variance, i.e., $\mathcal{CN}(0,1)$. Using the same Rician channel model, the channels from the transmit FA array to the $l$th RIS  $\mathbf{H}^{\rm b,r}_{l} \in \mathbb{C} ^ {N \times M}$ and the channel from the $l$th RIS to the $k$th UE $\mathbf{h}^{\rm r,u}_{k,l} \in \mathbb{C} ^ {M \times 1}$ are given by
\begin{align}
\mathbf{H}^{\rm b,r}_{l}  &\triangleq {\beta^{\rm b,r}_{l}}\left( {\sqrt {\frac{{{\kappa^{\rm b,r}_{l}}}}{{{\kappa^{\rm b,r}_{l}} + 1}}} {\bf{\bar H}}^{\rm b,r}_{l} + \sqrt {\frac{1}{{{\kappa^{\rm b,r}_{l}} + 1}}} {\bf{\tilde H}}^{\rm b,r}_{l}} \right),\label{H2}\\
\mathbf{h}^{\rm r,u}_{l,k} & \triangleq {\beta^{\rm r,u}_{l,k}}\left( {\sqrt {\frac{{{\kappa^{\rm r,u}_{l,k}}}}{{{\kappa^{\rm r,u}_{l,k}} + 1}}} {\bf{\bar h}}^{\rm r,u}_{l,k} + \sqrt {\frac{1}{{{\kappa^{\rm r,u}_{l,k}} + 1}}} {\bf{\tilde h}}^{\rm r,u}_{l,k}} \right),\label{H3}
\end{align}
where the LoS component ${\bf{\bar H}}^{\rm b,r}_{l}$ and ${\bf{\bar h}}^{\rm r,u}_{l,k}$ are given as ${\bf{\bar H}}^{\rm b,r}_{l} \triangleq  {{{\mathbf{\dot a}}}_N}(\theta^{\rm b,r}_l,\mathbf{z}) \mathbf{a}_M(\theta^{\rm b,r}_l,\phi^{\rm b,r}_l)^{\rm{H}}$ and ${\bf{\bar h}}^{\rm r,u}_{l,k} \triangleq \mathbf{a}_M(\theta^{\rm r,u}_{l,k},\phi^{\rm r,u}_{l,k})$, respectively. Also, the elements of the NLoS components ${\bf{\tilde H}}^{\rm b,r}_{k,l}$ and ${\bf{\tilde h}}^{\rm r,u}_{k,l}$ are i.i.d.~$\mathcal{CN}(0,1)$, and $\kappa^{\rm b,u}_{k} ,\  \kappa^{\rm b,r}_{l}, \ \kappa^{\rm r,u}_{l,k}$ are the Rician factors of each channel while ${\beta^{\rm b,u}_{k}},\ {\beta^{\rm b,r}_{l}},\ {\beta^{\rm r,u}_{l,k}}$ are the corresponding path losses in dB form, given by $\beta _{( \cdot )}^{( \cdot )} \triangleq  - {\beta _0} - 10\vartheta  _{( \cdot )}^{( \cdot )}{\log _{10}}\left( {r_{( \cdot )}^{( \cdot )}} \right)$, where $\vartheta  _{( \cdot )}^{( \cdot )}$ and ${r^{(\cdot)}_{(\cdot)}}$ denote the path loss exponent and the propagation distance of the corresponding channel. Finally, the overall cascaded channel ${{\mathbf{g}}_{k}} \in \mathbb{C}^{ N \times 1 }$ from the BS to the $k$th UE is formulated as
\begin{equation}
{{\mathbf{g}}_k} \triangleq {\mathbf{h}}_k^{{\text{b}},{\text{u}}} + \sum\limits_{l = 1}^L {{\mathbf{H}}_l^{{\text{b,r}}}{{\mathbf{E}}_l}{\mathbf{h}}_{l,k}^{{\text{r}},{\text{u}}}},
\end{equation}
where  ${\bf{E }}_l \triangleq {\rm{diag}} {( \mathbf{e}_l )} \in \mathbb{C}^{M \times M}$ denotes the phase shift matrix of the $l$th RIS, with ${\mathbf{e}}_l \triangleq {\left[ {{e^l_0},\;{e^l_1},\; \dots ,\;{e^l_{M - 1}}} \right]^{\rm{T}}}$ being the corresponding phase shift vector of the $l$th RIS constrained by $|e_m^l|=1, \forall m,l$. It is worth noting that due to the multiplicative accumulation of path loss in reflecting paths, it is reasonable to ignore channels with more than two reflections.

\vspace{-2mm}
\subsection{\color{black}Signal Model and Problem Formulation}
At the BS side, the input information bits are firstly modulated to a symbol vector $\mathbf{s} \triangleq [s_1,\ s_2, \dots, s_K]^{\rm T}$, which will further go through a digital and an analog beamformer denoted as $\mathbf{W} =[{{{\mathbf{w}}_1}}, \dots, {{{\mathbf{w}}_K}}] \in \mathbb{C}^{K\times K}$ and $\mathbf{V} \in \mathbb{C}^{N \times K}$. Note that the ABF is conducted by the analog phase shifter network so that the non-zero elements of $\mathbf{V}$ are constrained by the modulus-1 constraint, i.e., $| \left[{{{\mathbf{V}}}}\right]_{n,k} | = 1, \forall \left[{{{\mathbf{V}}}}\right]_{n,k} \neq 0$. Hence, the received signal at the $k$th UE is expressed as
\begin{equation}
{y_k} = {\mathbf{g}}_k^{\text{H}}{\mathbf{x}} + {\eta _k} = \underbrace {{\mathbf{g}}_k^{\text{H}}{\mathbf{V}}{{\mathbf{w}}_k}{s_k}}_{{\text{desired signal}}} + \underbrace {\sum\limits_{j \ne k}^K {{\mathbf{g}}_k^{\text{H}}{\mathbf{V}}{{\mathbf{w}}_j}{s_j}} }_{{\text{interference}}} + \underbrace {{\eta _k}}_{{\text{AWGN}}},
\end{equation}
where $\eta_{{k}} \sim \mathcal{CN}(0,\sigma_{\eta}^2) $ is the additive white Gaussian noise (AWGN) at the $k$th UE, with noise power being $\sigma_{\eta}^2$. 

Assuming identical noise levels across all UEs, the signal-to-interference-plus-noise ratio (SINR) of the $k$th UE is 
\begin{equation}
{\gamma _k} \triangleq {\left| {{\mathbf{g}}_k^{\text{H}}{\mathbf{V}}{{\mathbf{w}}_k}} \right|^2}/\left( {{{\sum\limits_{j \ne k}^K {\left| {{\mathbf{g}}_k^{\text{H}}{\mathbf{V}}{{\mathbf{w}}_j}} \right|^2} }} + \sigma _\eta ^2} \right).
\end{equation}
As such, the achievable sum rate over all the UEs is given by
\begin{equation}\label{sumrate}
R \triangleq \sum_{k=1}^K \log_2(1+\gamma_k).
\end{equation}

Our objective is to maximize the sum rate in (\ref{sumrate}). Considering the FA position limitations, the modulus-1 constraints, and the transmit power budget $P$, we aim to solve:
\begin{subequations}
\begin{align}
	&\mathop {\max }\limits_{\begin{subarray}{l} 
		{{\mathbf{E}}_l},{\mathbf{W}},{\mathbf{V}},{{\mathbf{z}}} 
	\end{subarray}} R \hfill \\
	&{{\text{s}}.{\text{t}}.}\;\left|   \left[ {{\mathbf{E}_l}} \right]_{m,m} \right| = 1,\;{\mkern 1mu} {\mkern 1mu} \forall m \in \mathcal{M}, \forall l \in \mathcal{L}, \hfill \\
	&\;\;\;\;\;\;\left|  \left[{{{\mathbf{V}}}}\right]_{n,k} \right| = 1,\;{\mkern 1mu} {\mkern 1mu} \forall \left[{{{\mathbf{V}}}}\right]_{n,k} \neq 0, \hfill \\
	&\;\;\;\;\;\;z_1 \geq 0,\;z_N \leq D, \hfill \\
	&\;\;\;\;\;\;z_{n + 1} - z_{n} \geq \delta, \; \forall n \in \mathcal{N}, \hfill \\
	&\;\;\;\;\;\;{\text{Tr}}\left( {{\mathbf{VW}}{{\mathbf{W}}^{\text{H}}}{{\mathbf{V}}^{\text{H}}}} \right) \leq P, \hfill 
\end{align}  
\label{opt1}%
\end{subequations}
in which $\mathcal{M} = \{1,2,\dots,M\}$, $\mathcal{N} = \{1,2,\dots,N\}$, $D$ and $\delta$ stand for, respectively, the maximum array aperture and minimum distance between adjacent FAs to avoid {\color{black}mutual coupling effect \cite{MCE}}. %Finally, $P$ is the maximum transmit power. 

\vspace{-2mm}
\section{Antenna Position and Joint Beamforming Optimization}\label{sec:opt}
Due to the coupling effects amongst the variables, solving (\ref{opt1}) is challenging. Moreover, the objective function (\ref{opt1}a) and constraints (\ref{opt1}b), (\ref{opt1}c) exhibit nonconvex features. To address these challenges, a BCD framework \cite{BCD} is used to decouple the variables. Then, the FP framework \cite{FP} will be employed to address the nonconvexity of (\ref{opt1}a). Lastly, the convexity of the resulting subproblems and the modulus-1 constraints (\ref{opt1}b), (\ref{opt1}c) will be ensured using the SDR method \cite{SDR}.

\vspace{-2mm}
\subsection{Optimization of Active HBF, $\mathbf{W}$ and $\mathbf{V}$}\label{BeamformingOptSec}
Here, a subproblem is formulated with respect to (w.r.t.) the DBF matrix $\mathbf{W}$ and ABF matrix $\mathbf{V}$ with fixed antenna position $\mathbf{z}$ and passive beamforming matrices $\mathbf{E}_l, \forall l \in \mathcal{L}$. 

Firstly, the quadratic transform in \cite{FP} is adopted to tackle the objective function with a sum-of-functions-of-ratio form. According to \cite[Corollary 2]{FP}, the objective function of the original problem can be formulated as
\begin{equation}
\mathop {{\text{max}}}\limits_{\mathbf{F},\bm{\alpha}} \sum\limits_{k = 1}^K {{f}} \left( {2{\alpha _k}\sqrt {{A_k}\left( {\mathbf{F}} \right)}  - \alpha _k^2{B_k}\left( {\mathbf{F}} \right)} \right),
\label{obj2}
\end{equation}
where $f\left( x \right) \triangleq {\log _2}\left( {1 + x} \right)$ is a concave and nondecreasing function, $\bm{\alpha} \triangleq [\alpha _1, \dots,\alpha_K]^{\text{T}}$ denotes an auxiliary variable vector, and ${\mathbf{F}} \triangleq [{\mathbf{f}}_1,\dots,{\mathbf{f}}_K]$, with ${{\mathbf{f}}_k} \triangleq {\mathbf{V}}{{\mathbf{w}}_k}$ denoting the overall beamforming matrix for the $k$th UE. Also, we have
%${A_k}\left( {\mathbf{F}} \right)$ and ${B_k}\left( {\mathbf{F}} \right)$ are respectively given by
\begin{align}
{A_k}\left( {\mathbf{F}} \right) &= {\left| {{\mathbf{g}}_k^{\text{H}}{\mathbf{V}}{{\mathbf{w}}_k}} \right|^2} = {\mathbf{w}}_k^{\text{H}}{{\mathbf{V}}^{\text{H}}}{\mathbf{g}}_k^{}{\mathbf{g}}_k^{\text{H}}{\mathbf{V}}{{\mathbf{w}}_k} \triangleq {\mathbf{f}}_k^{\text{H}}{\mathbf{R}}_k^g{\mathbf{f}}_k^{},\label{Ak}\\
{B_k}\left( {\mathbf{F}} \right) %&= \sum\limits_{j \ne k}^K {{\mathbf{w}}_j^{\text{H}}{{\mathbf{V}}^{\text{H}}}{\mathbf{g}}_k^{}{\mathbf{g}}_k^{\text{H}}{\mathbf{V}}{{\mathbf{w}}_j}}  + \sigma _\eta ^2 
&\triangleq \sum\limits_{j \ne k}^K {{\mathbf{f}}_j^{\text{H}}{\mathbf{R}}_k^g{\mathbf{f}}_j^{}}  + \sigma _\eta ^2.\label{Bk}
\end{align}

\subsubsection{Fully Digital (FD) Beamforming}\label{FDBF}
We first focus on the FD beamforming optimization, with the modulus-1 constraint (\ref{opt1}c) ignored temporarily. In addition, (\ref{opt1}b) and (\ref{opt1}e) are also removed temporarily since the antenna position at the BS side and passive beamforming at the RIS sides are fixed. Finally, with fixed $\bm{\alpha}$, (\ref{obj2}) is a concave function only if each ${{A_k}\left( {\mathbf{F}} \right)}$ is concave and each ${{B_k}\left( {\mathbf{F}} \right)}$ is convex. Letting ${{\mathbf{\Gamma }}_k} \triangleq {{\mathbf{f}}_k}{\mathbf{f}}_k^{\text{H}}$, the original problem can be formulated as  
\begin{subequations}
	\begin{align}
		&\mathop {{\text{max}}}\limits_{\begin{subarray}{l} 
				{{\mathbf{\Gamma }}_k},\; \bm{\alpha}  \\ 
				k \in \mathcal{K} 
		\end{subarray}}  \sum\limits_{k = 1}^K {f\left( {2{\alpha _k} \sqrt{{{\tilde A}_k}\left( {{{\mathbf{\Gamma }}_{{k}}}} \right)} - \alpha _k^2{{\tilde B}_k}\left( {{{\mathbf{\Gamma }}_{{k}}}} \right)} \right)}   \\
	&{{\text{s}}.{\text{t}}.}\; {\sum\limits_{k = 1}^K{\text{Tr}} \left( {{{\mathbf{\Gamma }}_k}}\right) } \leqslant P, \\
	&\;\;\;\;\;\;{{\mathbf{\Gamma }}_k} \succeq {\mathbf{0}},\;\forall k \in \mathcal{K}, \\
	&\;\;\;\;\;\;{\text{Rank}}\left( {{{\mathbf{\Gamma }}_k}} \right) = 1,\;\forall k \in \mathcal{K},
	\end{align}
	\label{opt2}%
\end{subequations}
where ${{\tilde A}_k}\left( {{{\mathbf{\Gamma }}_{{k}}}} \right) \triangleq  {{\text{Tr}}\left( {{\mathbf{R}}_k^g{{\mathbf{\Gamma }}_k}} \right)} $ and ${{\tilde B}_k}\left( {{{\mathbf{\Gamma }}_{{k}}}} \right) \triangleq \sum\limits_{j \ne k}^K {{\text{Tr}}\left( {{\mathbf{R}}_k^g{{\mathbf{\Gamma }}_j}} \right)}  + \sigma _\eta ^2$ both become affine functions w.r.t.~${{{\mathbf{\Gamma }}_k}}$. Furthermore, the newly introduced rank-1 constraint (\ref{opt2}d) is relaxed by SDR \cite{SDR}. Hence, the relaxed problem can be solved by CVX. With the optimal solutions ${\mathbf{\Gamma}^t_{k}}^{\star}$, auxiliary variables $\alpha_k$ can be updated in an alternating manner as
\begin{equation}
\alpha _k^{t} = \sqrt{{{\tilde A}_k}\left( {{{\mathbf{\Gamma }}_{{k}}}} \right)} / {{\tilde B}_k}\left( {{{\mathbf{\Gamma }}_{{k}}}} \right),
\label{alphacal}
\end{equation}
where the superscript $t$ denotes the iteration index of the FP. Therefore, we can obtain the optimal FD beamforming vector $\mathbf{f}^{\star}_k$ through eigenvalue decomposition (EVD) of $\mathbf{\Gamma}^{\star}_k$. %Finally, the tightness of the SDR is given in the following proposition. 

\begin{proposition}\label{prop1}
If the relaxed optimization problem (\ref{opt2}) is feasible, then there always exists an optimal solution ${\mathbf{\Gamma }}_{{\text{all}}}^ \star  = [{\mathbf{\Gamma }}_1^ \star \cdots {\mathbf{\Gamma }}_K^ \star]$, satisfying ${\text{rank}}\left( {{\mathbf{\Gamma }}_k^ \star } \right) = 1, \forall k \in \mathcal{K}$.
\end{proposition}

\begin{proof}
Please see Appendix \ref{AppdixA}.
\end{proof}

\subsubsection{Fully-connected (Full-Con) HBF}
As illustrated in \cite{HBFRef1}, the Full-Con HBF design can be achieved through a Euclidean distance minimization between the optimal FD beamforming matrix in $\mathbf{F}$ and Full-Con HBF matrix $\mathbf{V}\mathbf{W}$, formulated as
\begin{subequations}
	\begin{align}
	&\mathop {{\text{min}}}\limits_{\begin{subarray}{l} 
		{{\mathbf{W  }},\; \mathbf{V} } 	\end{subarray}} \left\| \mathbf{F}^\star - \mathbf{V}\mathbf{W} \right\|_F^2\\
	&{{\text{s}}.{\text{t}}.}\; \left|   \left[\mathbf{V}\right]_{i,j}  \right| = 1, \\
	&\;\;\;\;\;\;{\text{Tr}}\left( {{\mathbf{VW}}{{\mathbf{W}}^{\text{H}}}{{\mathbf{V}}^{\text{H}}}} \right) \leq P.\hfill 
	\end{align}
	\label{FC-HBF}%
\end{subequations}
This problem has already been well tackled in \cite{HBFRef2} through a manifold optimization. Although this problem is an approximation of the original problem, the explanation in \cite{HBFRef1,HBFRef2} and our simulation results illustrate that this approximation is reasonable and guarantees convergence.

\subsubsection{Sub-connected (Sub-Con) HBF}
Problem (\ref{opt2}) can be split into two subproblems w.r.t.~the DBF matrix $\mathbf{W}$ and ABF matrix $\mathbf{V}$. Given the ABF matrix $\mathbf{V}$, $A_k(\mathbf{w}_k)$ and $B_k(\mathbf{w}_k)$ can be derived as
\begin{align}
{A_k}\left( {{{\mathbf{w}}_k}} \right) &= {\mathbf{w}}_k^{\text{H}}{{\mathbf{V}}^{\text{H}}}{\mathbf{g}}_k^{}{\mathbf{g}}_k^{\text{H}}{\mathbf{V}}{{\mathbf{w}}_k} \triangleq {\mathbf{w}}_k^{\text{H}}{\mathbf{R}}_k^{Vg}{\mathbf{w}}_k^{},\label{Ak2}\\
{B_k}\left( {{{\mathbf{w}}_k}} \right) &\triangleq \sum\limits_{j \ne k}^K {{\mathbf{w}}_j^{\text{H}}{\mathbf{R}}_k^{Vg}{\mathbf{w}}_j^{}}  + \sigma _\eta ^2.\label{Bk2}
\end{align}
It can be observed that (\ref{Ak2}) and (\ref{Bk2}) have the same form as (\ref{Ak}) and (\ref{Bk}). Therefore, letting $\mathbf{\Omega}_k = \mathbf{w}_k \mathbf{w}_k^{\text{H}}$, the optimization subproblem of $\mathbf{\Omega}_k$ can be recast into 
\begin{subequations}
	\begin{align}
		&\mathop {{\text{max}}}\limits_{\begin{subarray}{l} 
				{{\mathbf{\Omega }}_k} \\ 
				k \in \mathcal{K} 
		\end{subarray}}  \sum\limits_{k = 1}^K {f\left( {2{\alpha _k} \sqrt{{{\tilde A}_k}\left( {{{\mathbf{\Omega }}_{{k}}}} \right)} - \alpha _k^2{{\tilde B}_k}\left( {{{\mathbf{\Omega }}_{{k}}}} \right)} \right)}   \\
		&{{\text{s}}.{\text{t}}.}\; {\sum\limits_{k = 1}^K {\text{Tr}} \left(\mathbf{V}^{\text{H}} \mathbf{V} {{{\mathbf{\Omega }}_k}} \right) } \leqslant P, \\
		&\;\;\;\;\;\;{{\mathbf{\Omega }}_k} \succeq {\mathbf{0}},\;\forall k \in \mathcal{K}, \\
		&\;\;\;\;\;\;{\text{Rank}}\left( {{{\mathbf{\Omega }}_k}} \right) = 1,\;\forall k \in \mathcal{K},
	\end{align}
	\label{optW}%
\end{subequations}
in which we have ${{\tilde A}_k}\left( {{{\mathbf{\Omega }}_k}} \right) \triangleq {\text{Tr}}\left( {{\mathbf{R}}_k^{Vg}{{\mathbf{\Omega }}_k}} \right)$ and ${{\tilde B}_k}\left( {{{\mathbf{\Omega }}_k}} \right) \triangleq \sum\limits_{j \ne k}^K {{\text{Tr}}\left( {{\mathbf{R}}_k^{Vg}{{\mathbf{\Omega }}_{j}}} \right)}  + \sigma _\eta ^2$. Since this subproblem has the same form as the FD one in (\ref{opt2}), it can be solved with the same method. With the optimal DBF solution ${\mathbf{w}_k^{t}}^{\star}$, we have 
\begin{equation}
{{\mathbf{w}_k^{t}}^{\star}}^{\text{H}} \mathbf{V}^{\text{H}} {\mathbf{R}}_k^g  \mathbf{V} {{\mathbf{w}_k^{t}}^{\star}} = \bm{ \nu }^{\text{H}} {\mathbf{\Lambda }}_k {\mathbf{R}}_k^g {\mathbf{\Lambda }}_k^{\text{H}} \bm{\nu},
\end{equation}
where $\bm{\nu}^{\text{H}} = [{\mathbf{v}}_1^{\text{H}} \cdots {\mathbf{v}}_K^{\text{H}}] \in \mathbb{C}^{1 \times N}$, with $\mathbf{V} = \text{blkdiag}[{\mathbf{v}}_1 \cdots {\mathbf{v}}_K]$, and $\mathbf{\Lambda}_k$ can be formulated as
\begin{equation}
	{\mathbf{\Lambda }}_k = {\text{blkdiag}}\left\{ {\begin{array}{@{\hspace{-1pt}}c@{\hskip 0.5pt}c@{\hskip 0.5pt}c@{\hspace{-1pt}}}
			{{\mathbf{w}_k^{t}}^{\star}}^{\text{H}}\left[1\right]{{{\mathbf{I}}_{N/K}}}\;\;& \cdots \;\;&	{{\mathbf{w}_k^{t}}^{\star}}^{\text{H}}\left[K\right]{{{\mathbf{I}}_{N/K}}} 
	\end{array}} \right\}.
\end{equation}
Letting $\mathbf{\Upsilon} = \bm{\nu} \bm{\nu}^{\text{H}}$, the subproblem w.r.t.~ABF can be cast as
\begin{subequations}
	\begin{align}
		&\mathop {{\text{max}}}\limits_{\begin{subarray}{l} 
				{{\mathbf{\Upsilon  }}} 
		\end{subarray}}  \sum\limits_{k = 1}^K {f\left( {2{\alpha _k} \sqrt{{{\tilde A}_k}\left( {{{\mathbf{\Upsilon }}}} \right)} - \alpha _k^2{{\tilde B}_k}\left( {{{\mathbf{\Upsilon }}}} \right)} \right)}   \\
		&{{\text{s}}.{\text{t}}.}\; {\sum\limits_{k = 1}^K {\text{Tr}} \left(\mathbf{\Lambda}_k \mathbf{\Lambda}_k^{\text{H}} {{{\mathbf{\Upsilon }}}} \right)}   \leqslant P, \\
		&\;\;\;\;\;\;\text{Tr}\left(\mathbf{O}_{n,n}\mathbf{\Upsilon}\right)=1, \; 1 \leq n \leq N, \\
		&\;\;\;\;\;\;{{\mathbf{\Upsilon }}} \succeq {\mathbf{0}}, \\
		&\;\;\;\;\;\;{\text{Rank}}\left( {{{\mathbf{\Upsilon }}}} \right) = 1,
	\end{align}
	\label{optV}%
\end{subequations}
where ${{\tilde B}_k}\left( {{{\mathbf{\Upsilon }}}} \right) \triangleq \sum\limits_{j \ne k}^K {{\text{Tr}}\left( { \mathbf{\Lambda}_j {\mathbf{R}}_k^{g} \mathbf{\Lambda}_j^{\text{H}} {{\mathbf{\Upsilon }}}} \right)}  + \sigma _\eta ^2$, ${{\tilde A}_k}\left( {{{\mathbf{\Upsilon }}}} \right) \triangleq {\text{Tr}}\left( { \mathbf{\Lambda}_k {\mathbf{R}}_k^{g} \mathbf{\Lambda}_k^{\text{H}} {{\mathbf{\Upsilon }}}} \right)$, and ${{{\mathbf{O}}_{n,n}}}$ is an $N \times N$ matrix with only the element in the $n$th row and $n$th column being $1$, and all other elements being $0$. Compared with subproblem (\ref{optW}), the number of decision vectors is reduced to $1$ from $K$, and extra equality constraints in (\ref{optV}c) are introduced. Although these changes do not destroy the convexity of this subproblem, they affect the tightness of the SDR, which means that we cannot get a perfect rank-1 solution after SDR. To tackle this problem, the rank-1 constraint (\ref{optV}e) is transformed into a more tractable form and added as a penalty term as $p\left(\text{Tr}\left(\mathbf{\Upsilon}\right) - \|\mathbf{\Upsilon}\|_2 \right) $, where $p<0$ is the penalty factor. Unfortunately, $p \|\mathbf{\Upsilon}\|_2$ is non-concave. However, this can be resolved by an iterative method \cite{Penalty1} or successive convex approximation (SCA) method \cite{Penalty2}, and the detailed process is omitted here for brevity. The optimal solution $\mathbf{\Upsilon}^{\star}$ with relatively good rank-1 property can be obtained with an appropriate penalty factor. Lastly, the update of $\bm{\alpha}$ follows a similar procedure as outlined in (\ref{alphacal}).

\subsection{Optimization of Passive Beamforming $\mathbf{E}_l$}\label{RISOptSec}
Here, we focus on the optimization of passive beamforming matrices $\mathbf{E}_l, \forall l \in \mathcal{L}$, with the optimal beamforming vector $\mathbf{W}$ and $\mathbf{V}$ obtained in Section \ref{BeamformingOptSec} along with a fixed antenna position $\mathbf{z}$. Letting ${\mathbf{\Xi }} \triangleq \left[{\mathbf{e}}_1^{\text{H}} ~\cdots~ {\mathbf{e}}_L^{\text{H}} ~1\right]^{\text{H}} \left[{\mathbf{e}}_1^{\text{H}} ~\cdots~ {\mathbf{e}}_L^{\text{H}} ~1\right],$ $A_k$ and $B_k$ in (\ref{Ak}) and (\ref{Bk}) can be rewritten w.r.t.~$\mathbf{\Xi }$ as
\begin{align}
{{\overset{\lower0.5em\hbox{$\smash{\scriptscriptstyle\frown}$}}{A} }_k}\left( {\mathbf{\Xi }} \right) &\triangleq {\text{Tr}}\left( {{\mathbf{\bar R}}_k^{\text{RIS}}{\mathbf{\Xi }}} \right),\\
{{\overset{\lower0.5em\hbox{$\smash{\scriptscriptstyle\frown}$}}{B} }_k}\left( {\mathbf{\Xi }} \right) &\triangleq {\text{Tr}}\left( {{\mathbf{\tilde R}}_k^{\text{RIS}}{\mathbf{\Xi }}} \right) + \sigma _\eta ^2,
\end{align}
where ${{\mathbf{\bar R}}_k^{\text{RIS}}}$ and ${{\mathbf{\tilde R}}_k^{\text{RIS}}}$ have the same structure, given by
\begin{equation}
\mathop {{\mathbf{R}}_k^{\text{RIS}}}\limits^{ - /\thicksim}  = \left[ {\begin{array}{*{20}{c}}
		{\mathop {{{\mathbf{A}}_{11}}}\limits^{ - /\thicksim} }& \cdots &{\mathop {{\mathbf{A}}_{L1}^{\text{H}}}\limits^{ - /\thicksim} }&{\mathop {{{\mathbf{B}_1}^{\text{H}}}}\limits^{ - /\thicksim} } \\  
		\vdots & \ddots & \vdots & \vdots   \\ 
		{\mathop {{{\mathbf{A}}_{L1}}}\limits^{ - /\thicksim} }& \cdots &{\mathop {{{\mathbf{A}}_{LL}}}\limits^{ - /\thicksim} }&{\mathop {{{\mathbf{B}_L}^{\text{H}}}}\limits^{ - /\thicksim} } \\ 
		{\mathop {\mathbf{B}_1}\limits^{ - /\thicksim} }& \cdots &{\mathop {\mathbf{B}_L}\limits^{ - /\thicksim} }&{\mathop a\limits^{ - /\thicksim} } 
\end{array}} \right],
\end{equation}
where each element is given by
\begin{align}
{{\mathbf{\bar A}}_{lr}} &= {\text{diag}}{\left( {{\mathbf{h}}_{l,k}^{{\text{r,u}}}} \right)^{\text{H}}}{{\mathbf{H}}{_{l}^{{\text{b,r}}}}^{\text{H}}}{{\mathbf{\Gamma }}_k}{\mathbf{H}}_{r}^{{\text{b,r}}}{\text{diag}}\left( {{\mathbf{h}}_{r,k}^{{\text{r,u}}}} \right),\\
{{{\mathbf{\tilde A}}}_{lr}} &= \sum\limits_{j \ne k} {{\text{diag}}{{\left( {{\mathbf{h}}_{l,k}^{{\text{r,u}}}} \right)}^{\text{H}}}{\mathbf{H}}{{_{l}^{{\text{b,r}}}}^{\text{H}}}{{\mathbf{\Gamma }}_j}{\mathbf{H}}_{r}^{{\text{b,r}}}{\text{diag}}\left( {{\mathbf{h}}_{r,k}^{{\text{r,u}}}} \right)},\\
{{{\mathbf{\bar B}}}_l} &= {\mathbf{h}}{_{k}^{{\text{b,u}}}}^{\text{H}}{{\mathbf{\Gamma }}_k}{\mathbf{H}}_{l}^{{\text{b,r}}}{\text{diag}}\left( {{\mathbf{h}}_{l,k}^{{\text{r,u}}}} \right),\\
{{{\mathbf{\tilde B}}}_l} &= \sum\limits_{j \ne k} {{\mathbf{h}}{{_{k}^{{\text{b,u}}}}}^{\text{H}}{{\mathbf{\Gamma }}_j}{\mathbf{H}}_{l}^{{\text{b,r}}}{\text{diag}}\left( {{\mathbf{h}}_{l,k}^{{\text{r,u}}}} \right)},\\
\bar a &= {{\mathbf{h}}{_{k}^{{\text{b,u}}}}^{\text{H}}}{{\mathbf{\Gamma }}_k}{\mathbf{h}}_{k}^{{\text{b,u}}}, \; \tilde a = \sum\limits_{j \ne k} {{\mathbf{h}}{{_{k}^{{\text{b,u}}}}^{\text{H}}}{{\mathbf{\Gamma }}_j}{\mathbf{h}}_{k}^{{\text{b,u}}}}.
\end{align}
Thus, a subproblem w.r.t.~$\mathbf{\Xi}$ can be formulated as 
\begin{subequations}
\begin{align}
&\max\limits_{\mathbf{\Xi}} \sum\limits_{k = 1}^K {f\left( {2{\alpha _k}\sqrt {{{\overset{\lower0.5em\hbox{$\smash{\scriptscriptstyle\frown}$}}{A} }_k}\left( {\mathbf{\Xi }} \right)}  - \alpha _k^2{{\overset{\lower0.5em\hbox{$\smash{\scriptscriptstyle\frown}$}}{B} }_k}\left( {\mathbf{\Xi }} \right)} \right)} \\
&{{\text{s}}.{\text{t}}.}\;{\text{Tr}}\left( {{{\mathbf{O}}_{m,m}}{\mathbf{\Xi }}} \right) = 1,\;1 \leqslant m \leqslant LM + 1,\\
&\;\;\;\;\;\; \mathbf{\Xi} \succeq \mathbf{0}\\
&\;\;\;\;\;\;{\text{rank}}\left( {\mathbf{\Xi }} \right) = 1,
\end{align}
\label{opt3}%
\end{subequations}
where ${{{\mathbf{O}}_{m,m}}}$ is an $LM + 1 \times LM + 1$ matrix with only the element in the $m$th row and $m$th column being $1$, and all other elements being $0$. Rank constraint (\ref{opt3}d) can be taken as a penalty function as in Problem (\ref{optV}) before. Finally, the auxiliary variable vector $\bm{\alpha}$ can be updated by
\begin{equation}\label{alphacale}
\alpha_k^t = \sqrt{{{\overset{\lower0.5em\hbox{$\smash{\scriptscriptstyle\frown}$}}{A} }_k}\left( {\mathbf{\Xi }} \right)} / {{\overset{\lower0.5em\hbox{$\smash{\scriptscriptstyle\frown}$}}{B} }_k}\left( {\mathbf{\Xi }} \right). 
%\alpha_k^t = \sqrt{{\stackrel{\frown}{A}}_k}\left({\mathbf{\Xi}}\right) / {\stackrel{\frown}{B}}_k\left({\mathbf{\Xi}}\right). 
\end{equation}

\vspace{-2mm}
\subsection{Optimization of Antenna Positions $\mathbf{z}$}\label{MAOptSec}
With the optimized active and passive beamforming in Sections \ref{BeamformingOptSec} and \ref{RISOptSec}, the antenna position $\mathbf{z}$ can be optimized in the following problem:
\begin{subequations}
\begin{align}
	&\mathop {\max }\limits_{\mathbf{z}} \sum\limits_{k = 1}^K {f\left( {2{\alpha _k}\sqrt {{{\overset{\lower0.5em\hbox{$\smash{\scriptscriptstyle\smile}$}}{A} }_k}\left( {\mathbf{z},\mathbf{f}_k} \right)}  - \alpha _k^2{{\overset{\lower0.5em\hbox{$\smash{\scriptscriptstyle\smile}$}}{B} }_k}\left( {\mathbf{z}} \right)} \right)} \\ 
	&{\text{s}}.{\text{t}}.\;\;{z_1} \geqslant 0,\;{z_N} \leqslant D, \hfill \\
	&\;\;\;\;\;\;{z_{n + 1}} - {z_n} \geqslant \delta ,\; 1 \leq n \leq N-1, \hfill 
\end{align}
\label{opt4}%
\end{subequations}
where ${{{\overset{\lower0.5em\hbox{$\smash{\scriptscriptstyle\smile}$}}{A} }_k}\left( {\mathbf{z},\mathbf{f}_k} \right)}$ is expressed as (\ref{Az}) (see top of this page), ${{\overset{\lower0.5em\hbox{$\smash{\scriptscriptstyle\smile}$}}{B} }_k}\left( {\mathbf{z}} \right) = \sum\limits_{j \ne k} {{{\overset{\lower0.5em\hbox{$\smash{\scriptscriptstyle\smile}$}}{A} }_k}\left( {\mathbf{z},\mathbf{f}_j} \right)}  + \sigma _\eta ^2$, ${\mathop{\bar\chi_k^{\text{b,u}}}}$, ${\mathop{\tilde\chi_k^{\text{b,u}}}}$, ${\mathop{\bar\chi_l^{\text{b,r}}}}$, and ${\mathop{\tilde\chi_l^{\text{b,r}}}}$ are the factor in front of the channel ${\mathop{\mathbf{\bar h}_k^{\text{b,u}}}}$, ${\mathop{\mathbf{\tilde h}_k^{\text{b,u}}}}$, ${\mathop{\mathbf{\bar H}_l^{\text{b,r}}}}$, and ${\mathop{\mathbf{\tilde H}_l^{\text{b,r}}}}$. Relevant expressions of some auxiliary variables and functions in (\ref{Az}) are listed in Appendix \ref{AppdixB}.
\begin{figure*}
\centering 
\begin{multline}\label{Az}
{{\overset{\lower0.5em\hbox{$\smash{\scriptscriptstyle\smile}$}}{A} }_k}\left( {{\mathbf{z}},{{\mathbf{f}}_k}} \right) =  {a_0}\left( {{{\bf{f}}_k}} \right) + {\left| {\bar \chi _k^{{\rm{b,u}}}} \right|^2}g_1^{\cos }\left( {\theta _k^{{\rm{b}},{\rm{u}}},\theta _k^{{\rm{b}},{\rm{u}}},{\bf{z}},{\bf{f}}_k^{}} \right) + 2\sum\limits_l {\Re \left\{ {\bar \chi _k^{{\rm{b,u}}}\bar \chi _l^{{\rm{b,r}}\;{\rm{*}}}{a_l}{g_1}\left( {\theta _l^{{\rm{b}},{\rm{r}}},\theta _k^{{\rm{b}},{\rm{u}}},{\bf{z}},{\bf{f}}_k^{}} \right)} \right\}} \\ 
+ \sum\limits_{{l_1}} {\sum\limits_{{l_2}} {\bar \chi _{{l_1}}^{{\rm{b,r}}}\bar \chi _{{l_2}}^{b,r\;*}a_{{l_1}}^*{a_{{l_2}}}{g_1}\left( {\theta _{{l_2}}^{{\rm{b}},{\rm{r}}},\theta _{{l_1}}^{{\rm{b}},{\rm{r}}},{\bf{z}},{\bf{f}}_k^{}} \right)} } 	+   {g_2}\left( {\theta _k^{{\text{b}},{\text{u}}},{\mathbf{z}},{\mathbf{b}}_k^{\text{H}}\left( {{\mathbf{f}}_k^{}} \right)} \right) + \sum\limits_l {{g_2}\left( {\theta _l^{{\text{b}},{\text{r}}},{\mathbf{z}},{\mathbf{c}}_{k,l}^{\text{H}}\left( {{\mathbf{f}}_k^{}} \right)} \right)}
\end{multline}
\hrulefill
\vspace{-3mm}
\end{figure*}

Although ${{{\overset{\lower0.5em\hbox{$\smash{\scriptscriptstyle\smile}$}}{A} }_k}\left( {\mathbf{z}},\mathbf{f}_k \right)}$ and ${{{\overset{\lower0.5em\hbox{$\smash{\scriptscriptstyle\smile}$}}{B} }_k}\left( {\mathbf{z}},\mathbf{f}_k \right)}$ appear to be very complicated, it can be discerned that they are sums of trigonometric functions w.r.t.~the elements in $\mathbf{z}$. Therefore, Problem (\ref{opt4}) is obviously nonconvex. {\color{black}In \cite{MAAnalogBF1}, the successive convex approximation (SCA) method based on the first-order Taylor expansion of trigonometric functions was employed to address the nonconvex beamforming gain constraint concerning the antenna position.} {\color{black}Similarly, when these nonconvex expressions appear in the objective function, MM framework \cite{MM} can be adopted here to solve this optimization problem.} Specifically, we can construct quadratic surrogate functions of the trigonometric functions through second-order Taylor expansion as
\begin{equation}\label{cosfun}
\cos \left( x \right) \ge  q\left( {x|{x_0}} \right) \triangleq \cos \left( {{x_0}} \right) - \sin \left( {{x_0}} \right)\left( {x - {x_0}} \right) - \frac{1}{2}{\left( {x - {x_0}} \right)^2}.
 \end{equation}
In fact, this is a global lower bound of this trigonometric function for any $x$ and $x_0$. Similarly, we have
\begin{equation}\label{sinfun}
\sin \left( x \right) \geq p \left(x|x_0\right) \triangleq \sin \left( {x_0} \right) + \cos \left( {{x_0}} \right)\left( {x - {x_0}} \right) - \frac{1}{2}{\left( {x - {x_0}} \right)^2}.
\end{equation}
Due to the fact that $-\cos(x) = \cos(x+\pi)$ and $-\sin(x) = \sin(x+\pi)$, the global lower bound of $-\cos(x)$ and $-\sin(x)$ can be taken as $q\left( {x + \pi|{x_0 + \pi}} \right)$ and $p \left(x + \pi|x_0 + \pi\right)$. 

Thus, each term in ${{{\overset{\lower0.5em\hbox{$\smash{\scriptscriptstyle\smile}$}}{A} }_k}\left( {\mathbf{z},\mathbf{f}_k} \right)}$ and ${{{\overset{\lower0.5em\hbox{$\smash{\scriptscriptstyle\smile}$}}{B} }_k}\left( {\mathbf{z},\mathbf{f}_k} \right)}$ can be lower bounded by the corresponding surrogate function according to the signs before these trigonometric functions. Assuming all trigonometric functions are with positive factors, we have 
\begin{align}
&{g_1^{\cos/\sin}}\left( {{\theta _1},{\theta _2},{\mathbf{z}},{{\mathbf{f}}_k}} \right)\notag\\
&\geqslant \sum\limits_{n = 1}^N {\sum\limits_{m = 1}^N {\left| \left[ \mathbf{f}_k \right]_n \left[ \mathbf{f}_k \right]_m \right|q/p\left( {u\left( {{z_n},{z_m}} \right)|u\left( {z_n^i,z_m^i} \right)} \right)} },\\
&{g_2}\left( {\theta ,{\mathbf{z}},{{\mathbf{b}}^{\text{H}}}} \right) \geqslant 2\sum\limits_{n = 1}^N {\Re \left\{  \left[\mathbf{b}\right]_n  \right\}q\left( {{z_n}\hat \theta |z_n^i\hat \theta } \right)}\notag\\
&\quad\quad\quad\quad\quad\quad\quad\quad\quad+ \Im \left\{ {{\mathbf{b}}{{\left[  n \right]}}} \right\}p\left( {{z_n}\hat \theta |z_n^i\hat \theta } \right),
\end{align} 
where the superscript $i$ denotes the index of the MM iteration. According to (\ref{cosfun}) and (\ref{sinfun}), the right sides of the inequalities are quadratic functions w.r.t.~$z_n$ and $z_m$. Therefore, the lower bound of ${{{\overset{\lower0.5em\hbox{$\smash{\scriptscriptstyle\smile}$}}{A} }_k}\left( {\mathbf{z}},\mathbf{f}_k \right)}$ can be expressed as a quadratic function of the variable $\mathbf{z}$ in a matrix form as
\begin{equation}
{\overset{\lower0.5em\hbox{$\smash{\scriptscriptstyle\smile}$}}{A} _k}\left( {{\mathbf{z}},{{\mathbf{f}}_k}} \right) \geq \frac{1}{2} {{\mathbf{z}}^{\text{T}}}{\mathbf{\bar R}}_k^{z^i}{\mathbf{z}} + {\mathbf{\bar r}}{_k^{z^i}}^{\text{T}} {\mathbf{z}} + {\bar c}_k^{z^i},
\label{Az2}
\end{equation}
where specific expressions of ${\mathbf{\bar R}}_k^z$, ${\mathbf{r}}_k^z$, and $c_k^z$ are given in {\color{black}(\ref{Rz})-(\ref{cz}).} The relevant auxiliary variables are given in {\color{black}Appendix \ref{AppdixB}}. Similarly, the lower bound of $-{{{\overset{\lower0.5em\hbox{$\smash{\scriptscriptstyle\smile}$}}{B} }_k}\left( {\mathbf{z}},\mathbf{f}_k \right)}$ can be expressed in a quadratic form as
\begin{equation}
-{{{\overset{\lower0.5em\hbox{$\smash{\scriptscriptstyle\smile}$}}{B} }_k}\left( {\mathbf{z}},\mathbf{f}_k \right)} \geq \frac{1}{2} {{\mathbf{z}}^{\text{T}}}{\mathbf{\tilde R}}_k^{z^i}{\mathbf{z}} + {\mathbf{ \tilde r}}{_k^{z^i}}^{\text{T}} {\mathbf{z}} + {\tilde c}_k^{z^i}.
\end{equation}
Since all quadratic coefficients in (\ref{cosfun}) and (\ref{sinfun}) are negative, ${\mathbf{\bar R}}_k^{z^i}$ and ${\mathbf{\tilde R}}_k^{z^i}$ have semi-negative properties. Therefore, Problem (\ref{opt4}) is a concave problem, and the optimal solution $\mathbf{z}^{\star}$ can be achieved as the expansion point $\mathbf{z}^{i+1}$ of next iteration until convergence. During this period, we update $\bm{\alpha}$ as
\begin{equation}\label{alphacalz}
\alpha_k^t = \sqrt{{{\overset{\lower0.5em\hbox{$\smash{\scriptscriptstyle\smile}$}}{A} }_k}\left( {\mathbf{z}} \right)}/{{\overset{\lower0.5em\hbox{$\smash{\scriptscriptstyle\smile}$}}{B} }_k}\left( {\mathbf{z}} \right).
\end{equation}
The details of the MM algorithm are given in Algorithm \ref{MMAlg}.

\vspace{-2mm}
{\color{black}
\begin{algorithm}[]
	\color{black}
\caption{MM Algorithm of Problem (\ref{opt4}).}\label{MMAlg}
\begin{algorithmic}[1]
		\STATE {\bf Initialize} Index: $t_3= 0$, $i=1$.  Initial feasible solution: $\mathbf{z}^0$.
		\REPEAT
		 \STATE Given overall active and passive beamformer: $\mathbf{F}$ and $\mathbf{e}_{1:L}$  as well as expansion point $\mathbf{z}^i$, calculate $ {\mathbf{\bar R}}_k^{z^i}$, ${\mathbf{\bar r}}{_k^{z^i}}^{\text{T}}$, ${\bar c}_k^{z^i}$ and $ {\mathbf{\tilde R}}_k^{z^i}$, ${\mathbf{\tilde r}}{_k^{z^i}}^{\text{T}}$, ${\tilde c}_k^{z^i}$ in {\color{black}Appendix \ref{AppdixB}}; \\
		\REPEAT
			\STATE update $\bm{\alpha}^{t_3}$ according to (\ref{alphacalz}); \\
			\STATE Calculate the optimal APV ${{\tilde{\bf{z}}}^{t_3}}$ according to (\ref{opt4});\\
			\STATE $t_3 \leftarrow t_3+1$.
		\UNTIL{$\frac{{{\text{Obj}}\left( {{\mathbf{F}},{{\mathbf{e}}_{1:L}},{{{\mathbf{\tilde z}}}^{t_3 - 1}}} \right)}}{{{\text{Obj}}\left( {{\mathbf{F}},{{\mathbf{e}}_{1:L}},{{{\mathbf{\tilde z}}}^{t_3 - 2}}} \right)}} \leq  \rho $}\\
		\STATE $\mathbf{z}^i = {{\tilde{\bf{z}}}^{t_3-1}}$, ${{\bf{z}}^{\star}} = \mathbf{z}^i$; \\
		\UNTIL{$\frac{{{\text{Obj}}\left( {{\mathbf{F}},{{\mathbf{e}}_{1:L}},{{\mathbf{z}}^{i - 1}}} \right)}}{{{\text{Obj}}\left( {{\mathbf{F}},{{\mathbf{e}}_{1:L}},{{\mathbf{z}}^{i - 2}}} \right)}} \leq \rho$ or $i-1 > I_{\rm MM}^{\rm max}$ }.\\
		\STATE $I_{\rm MM} = i-1$.\\
		
				\ENSURE ~~\\
		Optimal APV $\mathbf{z}^{\star}$;\\
\end{algorithmic}
\end{algorithm}
}

\vspace{-4mm}
\subsection{Overall Algorithm and Complexity of Problem (\ref{opt1})}
\begin{figure}[ht]
	\centering
	\includegraphics[scale=0.6]{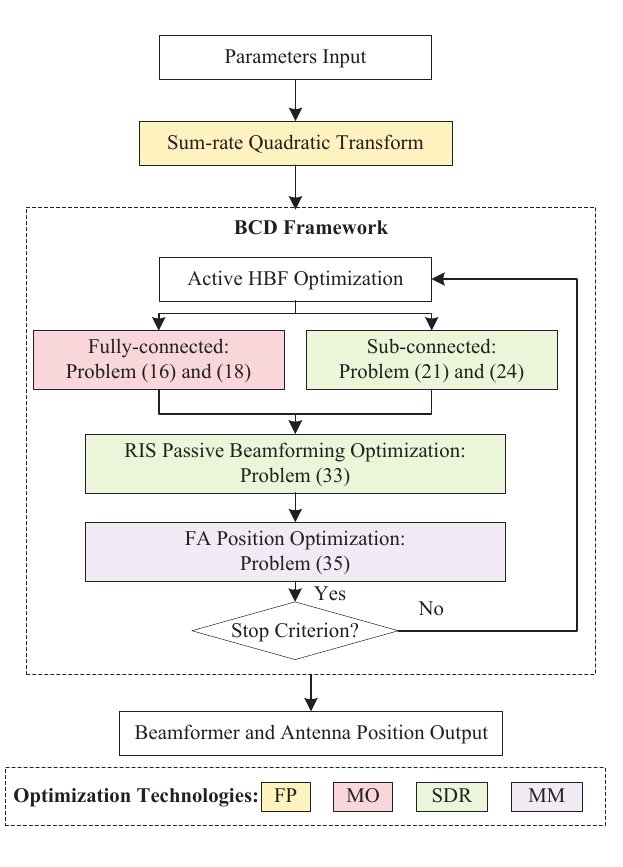}
	\caption{\color{black}Optimization flow chart of the proposed RIS-assisted FAS.}\label{FlowChart1}
	\vspace{-4mm}
\end{figure}
Based on the above three subsections, the detailed descriptions of Problem (\ref{opt1}) are summarized in Algorithm \ref{BCDAlg}, {\color{black}where $\rho$ is the convergence threshold.} {\color{black} In additional, a flow chart is provided in Fig. \ref{FlowChart1} to clearly illustrate the key steps of the algorithm. }

\subsubsection{Convergence Analysis}
From Step 3 to 7, we optimize the active beamforming, including fully digital and hybrid ones. The fully digital beamformer is optimal using the SDR method, while the alternating optimization between DBF and ABF in HBF architecture ensure the objective function rises until convergence. From Step 9 to 13, we optimize the phase shift matrix of RISs, similar to the ABF optimization. Finally, at Step 15, we optimize the antenna position under the MM framework. Its convergence is illustrated in (\ref{MMConvergenceEq}), where inequality (a) holds since Problem (\ref{opt4}) is non-decreasing. As a result, objective function value ${\rm{Obj}}\left( {{{\bf{z}}^{i - 1}}|{{\bf{W}}^i},{{\bf{V}}^i},{\bf{e}}_{1:L}^i,{{\bf{z}}^{i - 1}}} \right) \le {\rm{Obj}}\left( {{{\bf{z}}^i}|{{\bf{W}}^i},{{\bf{V}}^i},{\bf{e}}_{1:L}^i,{{\bf{z}}^{i - 1}}} \right)$. On the other hand, we have ${\rm{Obj}}\left( {{{\bf{z}}^i}|{{\bf{W}}^i},{{\bf{V}}^i},{\bf{e}}_{1:L}^i} \right)- {\rm{Obj}}\left( {{{\bf{z}}^i}|{{\bf{W}}^i},{{\bf{V}}^i},{\bf{e}}_{1:L}^i,{{\bf{z}}^{i - 1}}} \right) \ge 0$ as the second term is a global lower bound surrogate function of the first term. As such, monotonic convergences of all three subproblems are guaranteed, and Algorithm \ref{BCDAlg} converges.

\subsubsection{Complexity Analysis}
The optimizations of active and passive beamforming at the BS and RISs end up as convex semidefinite programming problems, which can be solved by the interior point method with a complexity order of ${\cal O}\left( {\max {{\left( {m,n} \right)}^4}{n^{1/2}}\log \left( {1/\varepsilon } \right)} \right)$ \cite{SDR}, where $n$ is the dimension of the optimization variable, $m$ is the number of equality and inequality constraints, and $\varepsilon$ is the solution accuracy. Also, the optimization of antenna position is a convex quadratic programming, whose complexity order is ${\cal O}\left( { {{\left( n+m \right)}^{1.5}}{n}\log \left( {1/\varepsilon } \right)} \right)$ \cite{CVX}. Specifically, at Step 5, the complexity order of Problem (\ref{optW}) and (\ref{optV}) is  ${\cal O}\left( \left(K^{4.5}+N^{4.5}\right)\log \left( {1/\varepsilon } \right)\right)$ for Sub-Con HBF. Similarly, the complexity order of (\ref{opt3}) is ${\cal O} \left(\left(LM\right)^{4.5}\log \left( {1/\varepsilon } \right)\right)$. Finally, the complexity order of each iteration in Problem (\ref{opt4}) is  ${\cal O}\left( { {{ N }^{2.5}}\log \left( {1/\varepsilon } \right)} \right)$. Assuming the number of times of MM operations is $I_{\text{MM}}$, the overall complexity order of one iteration of the BCD is given by
\begin{equation}
\small
{\cal O}\left\{ {\max \left[ {{{\left( t_1 {N}^{4.5} + K^{4.5} \right)}}, t_2{{\left( {LM} \right)}^{4.5}}, t_3{I_{{\text{MM}}}}{N^{2.5}}} \right]\log \left( {\frac{1}{\varepsilon }} \right)} \right\}.
\notag
\end{equation}
%\vspace{-0.25cm}

\begin{figure*}[hb]
\centering
\hrulefill
\begin{equation}
	\begin{array}{l}
		{\rm{Obj}}\left( {{{\bf{z}}^{i - 1}}|{{\bf{W}}^i},{{\bf{V}}^i},{\bf{e}}_{1:L}^i} \right) = {\rm{Obj}}\left( {{{\bf{z}}^{i - 1}}|{{\bf{W}}^i},{{\bf{V}}^i},{\bf{e}}_{1:L}^i,{{\bf{z}}^{i - 1}}} \right) + {\rm{Obj}}\left( {{{\bf{z}}^{i - 1}}|{{\bf{W}}^i},{{\bf{V}}^i},{\bf{e}}_{1:L}^i} \right) - {\rm{Obj}}\left( {{{\bf{z}}^{i - 1}}|{{\bf{W}}^i},{{\bf{V}}^i},{\bf{e}}_{1:L}^i,{{\bf{z}}^{i - 1}}} \right)\\
		\mathop  \le \limits^{\left( a \right)} {\rm{Obj}}\left( {{{\bf{z}}^i}|{{\bf{W}}^i},{{\bf{V}}^i},{\bf{e}}_{1:L}^i,{{\bf{z}}^{i - 1}}} \right) + {\rm{Obj}}\left( {{{\bf{z}}^i}|{{\bf{W}}^i},{{\bf{V}}^i},{\bf{e}}_{1:L}^i} \right) - {\rm{Obj}}\left( {{{\bf{z}}^i}|{{\bf{W}}^i},{{\bf{V}}^i},{\bf{e}}_{1:L}^i,{{\bf{z}}^{i - 1}}} \right) = {\rm{Obj}}\left( {{{\bf{z}}^i}|{{\bf{W}}^i},{{\bf{V}}^i},{\bf{e}}_{1:L}^i} \right)
	\end{array}
	\label{MMConvergenceEq}
\end{equation}
\end{figure*}
{\color{black}
\vspace{-2mm}
\begin{algorithm}[]
	\color{black}
\caption{BCD Optimization of Problem (\ref{opt1}).}\label{BCDAlg}
	\begin{algorithmic}[1]
		\STATE {\bf Initialize} Index: $\mathbf{z}^0$, $i=1$, $t_1 = t_2 = 0$. Initial feasible solution: $\{\mathbf{W}^0,\mathbf{V}^0\}$, $\{\mathbf{e}_1^0, \dots , \mathbf{e}_L^0\}$.
		\REPEAT
		\REPEAT
		\STATE Given $\{{\bf{e}}_{1:L}^{i - 1}, {{\bf{z}}^{i - 1}}\}$, update $\bm{\alpha}^{t_1}$ according to (\ref{alphacal}); \\
		\STATE Given $\{{\bf{e}}_{1:L}^{i - 1}, {{\bf{z}}^{i - 1}}\}$, calculate the optimal HBF matrices  $\tilde{\mathbf{W}}^{t_1}$, $\tilde{\mathbf{V}}^{t_1}$ according to Probs. (\ref{opt2}, \ref{FC-HBF}, \ref{optW}, \ref{optV});\\
		\STATE $t_1 \leftarrow t_1 + 1$.
		\UNTIL{$\frac{{\rm Obj}\left(\tilde{\mathbf{W}}^{t_1-1},\tilde{\mathbf{V}}^{t_1-1},\mathbf{e}_{1:L}^{i-1},\mathbf{z}^{i-1}\right)}{{\rm Obj}\left(\tilde{\mathbf{W}}^{t_1-2},\tilde{\mathbf{V}}^{t_1-2},\mathbf{e}_{1:L}^{i-1},\mathbf{z}^{i-1}\right)}\leq \rho$}.\\
		\STATE $\mathbf{W}^i = \tilde{\mathbf{W}}^{t_1-1}$, $\mathbf{V}^i = \tilde{\mathbf{V}}^{t_1-1}$.\\
		\REPEAT
		\STATE Given $\{\mathbf{W}^i, \mathbf{V}^i,{{\bf{z}}^{i - 1}}\}$, update $\bm{\alpha}^{t_2}$ according to (\ref{alphacale}); \\
		\STATE Given $\{\mathbf{W}^i, \mathbf{V}^i,{{\bf{z}}^{i - 1}}\}$, calculate the optimal RIS phase shift vector ${\tilde{\bf{e}}}_{1:L}^{t_2}$ according to Prob. (\ref{opt3});\\
		\STATE  $t_2 \leftarrow t_2 + 1$.
		\UNTIL{$\frac{{\rm Obj}\left({\mathbf{W}}^{i},{\mathbf{V}}^{i},\tilde{\mathbf{e}}_{1:L}^{t_2-1},\mathbf{z}^{i-1}\right)}{{\rm Obj}\left({\mathbf{W}}^{i},{\mathbf{V}}^{i},\tilde{\mathbf{e}}_{1:L}^{t_2-2},\mathbf{z}^{i-1}\right)}\leq \rho$}.\\
		\STATE $\mathbf{e}_{1:L}^{i} = \mathbf{e}_{1:L}^{t_2-1}$.\\
		\STATE Calculate the optimal APV $\mathbf{z}^i$ through Algorithm \ref{MMAlg}; \\
		\STATE $\mathbf{W}^{\star} = \mathbf{W}^{i}$, $\mathbf{V}^{\star} = \mathbf{V}^{i}$;\\
		\STATE $\{\mathbf{e}_1^{\star}, \; \dots , \; \mathbf{e}_L^{\star}\} = \{\mathbf{e}_1^{i}, \; \dots , \; \mathbf{e}_L^{i}\}$;\\
		\STATE ${{\bf{z}}^{\star}} = \mathbf{z}^i$;\\
		\STATE $i \leftarrow i + 1$.
		\UNTIL{$\frac{{{\text{Obj}}\left( {{{\mathbf{W}}^{i-1}},{{\mathbf{V}}^{i-1}},{\mathbf{e}}_{1:L}^{i-1},{{\mathbf{z}}^{i-1}}} \right)}}{{{\text{Obj}}\left( {{{\mathbf{W}}^{i - 2}},{{\mathbf{V}}^{i - 2}},{\mathbf{e}}_{1:L}^{i - 2},{{\mathbf{z}}^{i - 2}}} \right)}} \leq \rho $ or $i - 1 > I_{\rm BCD}^{\rm max}$.}
		\STATE $I_{\rm BCD} = i-1$.
		\ENSURE ~~\\
		Optimal beamforming matrix:  $ \{\mathbf{W}^{\star} ,\; \mathbf{V}^{\star}\}$;\\
		Optimal RIS phase shift vector:$\{\mathbf{e}_1^{\star}, \; \dots , \; \mathbf{e}_L^{\star}\}$;\\
		Optimal APV: $\mathbf{z}^{\star}$.
		
\end{algorithmic}
\end{algorithm}
}
\vspace{-3mm}
\section{\color{black}Low-Complexity TFA with Multiple Delicately Placed RISs}\label{sec:tfa}
In this section, we propose a novel type of position switching, dubbed telescopic movement, where the spacing between neighboring antennas is consistent but adjustable. A TFA array architecture that realizes this concept is depicted in Fig.~\ref{TelescopicFA}. Controlling the diamond-shaped track with a motor enables the array elements to move telescopically. Next, each subarray at the BS will be replaced by the proposed TFA array, and a low-complexity Sub-Con HBF scheme will be presented in the context of a LoS-dominant channel condition {\color{black}and multiple RISs with delicate deployments}.

\begin{figure}[ht]
\centering
\includegraphics[scale=0.4]{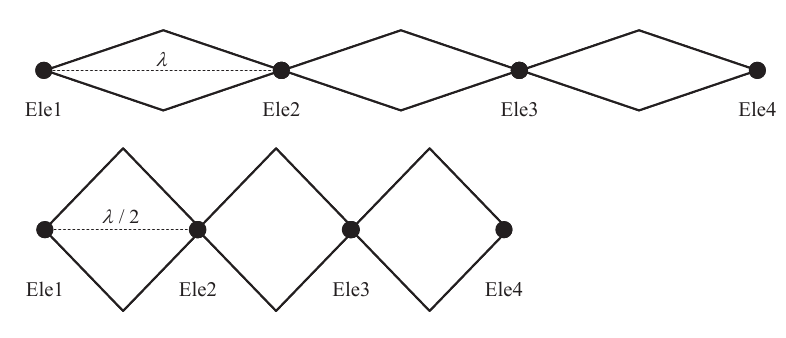}
\caption{Architecture and movement pattern illustration of a $4$-element TFA.}\label{TelescopicFA}
\vspace{-4mm}
\end{figure}

\vspace{-4mm}
\subsection{GL Effect}
In \cite{FAPortSelRIS}, a single-RIS-aided single-user MISO TFA was proposed based on the GL effect. Here, we extend it into our model. First, we rewrite the steering vector for each TFA subarray as
\begin{equation}\label{TFASV}
{\mathbf{\dot b}}_{\dot N}\left( {\theta ,\dot d} \right) \triangleq {\left[ {1,\;{e^{ - \jmath 2\pi \dot d\cos \left( \theta  \right)}},\; \dots ,\;{e^{ - \jmath 2\pi \left( {\dot N - 1} \right)\dot d\cos \left( \theta  \right)}}} \right]^{\text{T}}},
\end{equation}
where $\dot N = N/K$. We assume that all $K$ subarrays share the same observation angle and distance for all UEs in the far-field, and the phase differences between different subarrays are ignored. The relationship between element spacing $\dot d$, GL index $k$, GL angle $\theta_g^k$, and ML angle $\theta_m$  is given by \cite{FAPortSelRIS}
\begin{equation}\label{GLSpacing}
\dot d(\theta_g^k,\theta_m) = \frac{{k\lambda }}{{\cos (\theta_g^k ) - \cos ({\theta _m})}},\;k =  \pm 1,\; \pm 2,\; \cdots .
\end{equation}
%According to (\ref{GLSpacing}), we find that the presence of GLs are traceable, as stated in Prop. \ref{prop2}. However, some limitations of  GL controls are  unavoidable, as discussed in Prop. \ref{prop3}.

\begin{figure}[ht]
\centering
\includegraphics[scale=0.8]{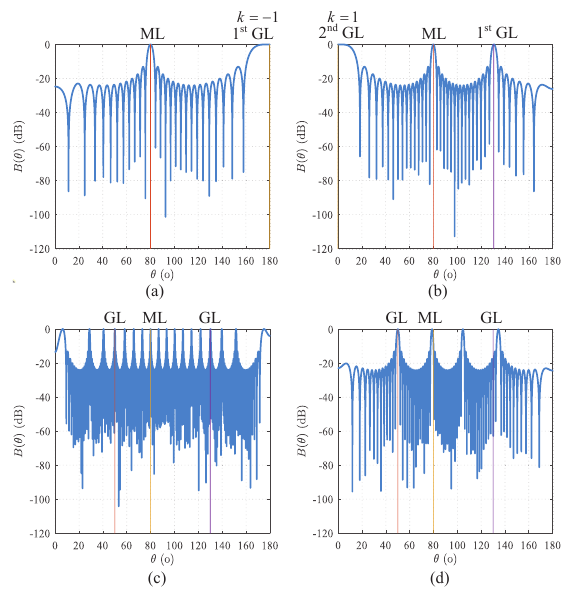}
\caption{Beampattern of a $16$-element ULA with different element spacing when $\theta_m = 80^{\text{o}}$: (a) $\theta_g^{-1} = 180^{\text{o}}$; (b) $\theta_g^1 = 0^{\text{o}}$; (c) $\theta_g^{k1} = 50^{\text{o}}$, $\theta_g^{k2}= 1$ $ 30^{\text{o}}$ with approximation tolerance $0.01$; (d) $\theta_g^{k1} = 50^{\text{o}}$, $\theta_g^{k2} = 130^{\text{o}} $ with approximation tolerance $0.1$.}\label{FABeamPattern}
\vspace{-3mm}
\label{SpatialInt}
\end{figure}

\begin{proposition}\label{prop2}
Assuming $\theta_m < \pi/2$, new GLs appear at the spacing of ${\dot{d}}(\pi, \theta_m)$ for $\forall k < 0$, and $\dot{d}(0,\theta_m)$ for $\forall k>0$.
\end{proposition}

\begin{proof}
 $\cos(0) - \cos(\theta_m)$ is the maximal positive number for any given $\theta_m$. By contrast, $\cos(\pi) - \cos(\theta_m)$ is the minimal negative number. Thus, upon increasing spacing $\dot{d}$, GLs with positive and negative indices appear from $0$ and $\pi$ angles, respectively. The process is shown in Figs.~\ref{FABeamPattern}(a) and \ref{FABeamPattern}(b).
\end{proof}

\begin{proposition}\label{prop3}
Completely free control over the number and angle of GL is not recommended, as it may lead to severe spatial-spectral leakage.
\end{proposition}

\begin{proof}
Let the angles of the desired ML and two GLs be denoted by $\theta_m$, $\theta_g^1$, and $\theta_g^2$, respectively. According to (\ref{GLSpacing}), the following equation set is established as
\begin{equation}\label{GLEqSet}
\left\{ \begin{gathered}
\dot d = \frac{{{k_1}\lambda }}{{\cos (\theta _g^1) - \cos ({\theta _m})}}, \hfill \\
\dot d = \frac{{{k_2}\lambda }}{{\cos (\theta _g^2) - \cos ({\theta _m})}}. \hfill \\
\end{gathered}  \right.
\end{equation}
Let $\frac{{\cos (\theta _g^2) - \cos ({\theta _m})}}{{\cos (\theta _g^1) - \cos ({\theta _m})}}$ be a rational number, expressed in its simplest fractional form as $p/q$, where $p$ and $q$ are coprime integers sharing the same sign as ${\cos (\theta _g^2) - \cos ({\theta _m})}$ and ${\cos (\theta _g^1) - \cos ({\theta _m})}$, respectively. Therefore, we have
\begin{equation}\label{MultiGLSolu}
\frac{k_1}{k_2} = \frac{{\cos (\theta _g^2) - \cos ({\theta _m})}}{{\cos (\theta _g^1) - \cos ({\theta _m})}}  = \frac{p}{q}.
\end{equation}
With $k_1$ and $k_2$, $\dot d$ can be obtained by (\ref{GLEqSet}). However, the problem is that the values of $p$ and $q$ in (\ref{MultiGLSolu}) may be very large, leading to a large $\dot d$ and massive GLs in the spatial-spectrum shown in Fig.~\ref{FABeamPattern}(c). Although a finite bit approximation of the ratio number can mitigate this issue, it also introduces beam pointing errors, as illustrated in Fig.~\ref{FABeamPattern}(d).
\end{proof}

Fortunately, in certain special cases, completely free control of GL can be achieved. For example, assuming $\theta_m < \pi/2$, completely free control of one GL can be achieved in the angle range of $[\pi/2,\pi]$ \cite{FAPortSelRIS}. {\color{black}This limitation implies that we can only ensure coverage for UEs located on the opposite side of the RIS, significantly restricting the practicality of this scheme. To address this issue, we deploy multiple RISs positioned on either side of the transmitting array to provide coverage for UEs across the entire space\footnote{\color{black}The proposed algorithm remains applicable to single-RIS configurations, though this may induce spatial neighborhood interference as demonstrated in Fig. \ref{SpatialInt}. Additionally, future work will explore rotational mobility for TFA subarrays, where each array's normal vector dynamically aligns between the RIS and served UEs to reduce spatial interference.}.} Furthermore, we set the maximum telescopic spacing to $\dot d_{\max} = \lambda$, which yields only one ML and one GL to prevent severe spatial-spectral leakage. In this sense, we can generate one ML and one GL towards UE and RIS, respectively, if they are located at the different side of the BS, i.e. satisfying $\theta^{\text{b,u}} < \pi/2$, $\theta^{\text{b,r}} > \pi/2$ or $\theta^{\text{b,u}} > \pi/2$, $\theta^{\text{b,r}} < \pi/2$. Therefore, we classify the UEs and RISs on different sides of the BS into one group.

%\vspace{-0.2cm}
\subsection{Joint Design of ABF and Antenna Position}\label{TFABF}
In this part, all $K$ TFA subarrays are employed to serve $K$ UEs individually. For the $k$th UE, the ML is designed towards the direction of this UE $\theta^{\text{b,u}}_k$. While the GL is designed to point to the corresponding RIS, denoted as $\theta_{\xi (k)}^{\text{b,r}}$, where $\xi (k)$ returns the index of RIS opposite to the $k$th UE. Consequently, the spacing of the $k$th TFA subarray can be derived as
\begin{equation}\label{TFAPos}
\dot d_k = \frac{{ \lambda }}{\left|{\cos (\theta_{\xi (k)}^{\text{b,r}} ) - \cos ({\theta_k^{\text{b,u}}})}\right|}.
\end{equation}

With this spacing, we have ${{{\mathbf{\dot b}}}_{\dot N}}\left( {\theta _k^{{\text{b,u}}},\dot d} \right) = {{{\mathbf{\dot b}}}_{\dot N}}\left( {\theta _{\xi (k)}^{{\text{b,r}}},\dot d} \right)$, and the corresponding ABF vector $\mathbf{v}_k$ can be chosen as
\begin{equation}\label{TFAABF}
\mathbf{v}_k = {{{\mathbf{\dot b}}}_{\dot N}}\left( {\theta _k^{{\text{b,u}}},\dot d_k} \right) = {{{\mathbf{\dot b}}}_{\dot N}}\left( {\theta _{\xi (k)}^{{\text{b,r}}},\dot d_k} \right),
\end{equation}
which satisfies the modulus-1 constraint, and can steer beams with full array gain towards both the UE and RIS. Since we design $K$ TFA subarrays independently, there exist inter-user and inter-RIS interference, to be addressed by DBF.

\vspace{-2mm}
\subsection{Low-Complexity Design of Passive Beamforming at RIS}
With the FA position and ABF designs in the last subsection, a straightforward way to obtain the RIS phase shift matrix is using the technique in Section \ref{RISOptSec}. Instead of the joint optimization of RIS in Section \ref{RISOptSec}, this subsection gives a suboptimal scheme, which optimizes each RIS individually for complexity reduction. Thanks to the quasi-orthogonal ABF design in Section \ref{TFABF}, we {\color{black}temporarily} ignore interference terms without any attribution of ML or GL. Thus, the reflected signal vector from the $l$th RIS to the UEs is expressed as
\begin{equation}\label{RecTFARIS}
\begin{aligned}
&{{\mathbf{y}}_{{\rm{RI}}{{\rm{S}}_{l}}}} = \frac{1}{{\sqrt N }}{\bf{H}}_{l}^{{\rm{r,u}}\;{\rm{H}}}{\bf{E}}_l^{\rm{H}}\sum\limits_{i \in \zeta \left( l \right)} {{\bf{\dot H}}_{l,i}^{{\rm{b,r}}\;{\kern 1pt} {\rm{H}}}{{\bf{v}}_i}{s_i}}  + {{\eta }} \\
& \mathop  =  \frac{{\bar \chi _l^{{\rm{b,r}}}\dot N\sum\limits_{i \in \zeta \left( l \right)} {{s_i}} }}{{\sqrt N }}{\bf{H}}_{l,k}^{{\rm{r,u}}\;{\rm{H}}}{\rm{diag}}\left[ {{{\bf{a}}_M}\left( {\theta _l^{{\rm{b,r}}},\phi _l^{{\rm{b,r}}}} \right)} \right]{\bf{e}}_l^* + {{\eta }},
\end{aligned}
\end{equation}
in which we have ${\bf{H}}_l^{{\rm{r,u H}}} = \left[{\bf h}_{l,1}^{{\rm r},{\rm u}} \cdots {\bf h}_{l,K}^{{\rm r},{\rm u}}\right]^{\rm{H}}$, and ${\bf{\dot H}}_{l,k}^{{\rm{b}},{\rm{r}}\;{\rm{H}}} = \bar \chi _l^{{\rm{b,r}}}{{\bf{a}}_M}\left( {\theta _l^{{\rm{b,r}}},\phi _l^{{\rm{b,r}}}} \right){\bf{\dot b}}^{\text{H}}\left( {{\theta _k},{{\dot d}_k}} \right)$. Also, $\zeta(l)$ denotes the inverse function of $\xi (k)$, which returns the indices of the TFA subarray generating GL towards the $l$th RIS. It is noteworthy that two assumptions are required to obtain (\ref{RecTFARIS}). The first one is LoS-dominant channel condition, i.e., $\kappa$ is large, while the second one is identical DBF matrix, i.e., $\mathbf{W} = \mathbf{I}_K/\sqrt{N}$. To maximize the received power of all served UEs with index $k \in \zeta(l)$ and minimize the leakage power to the other UEs with index $k \notin \zeta (l)$, the signal-to-leakage-and-noise ratio (SLNR) metric at the $l$th RIS can be defined as
\begin{equation}
{\rm{SLNR}}\left( {{{\bf{e}}_l}} \right) \triangleq \frac{{{\bf{e}}_l^{\rm{H}}\sum\limits_{k \in \zeta \left( l \right)} {{{\bf{h}}^{\text{b,r,u}}_{l,k}}{\bf{h}}^{\text{b,r,u H}}_{l,k}} {{\bf{e}}_l}}}{{{\bf{e}}_l^{\rm{H}}\left[ {\sum\limits_{j \notin \zeta \left( l \right)} {{{\bf{h}}^{\text{b,r,u}}_{l,j}}{\bf{h}}_{l,j}^{\text{b,r,u H}}}  + \frac{{\sigma _\eta ^2}}{M}{\bf{I}}} \right]{{\bf{e}}_l}}},
\end{equation}
where ${{\mathbf{h}}^{\text{b,r,u}}_{l, k}} \triangleq {\text{diag}}\left[ {{{\mathbf{a}}_M}\left( {\theta _l^{{\text{b,r}}},\phi _l^{{\text{b,r}}}} \right)} \right]{\mathbf{h}}_{l,k}^{{\text{r,u}}\;{*}}$. The SLNR can be maximized through the generalized Rayleigh-Ritz theorem. The solution can be obtained from the eigenvector corresponding to the largest eigenvalue of the matrix
\begin{equation}\label{RISRR}
{\left[ {\sum\limits_{j \notin \zeta \left( l \right)} {{{\bf{h}}^{\rm b,r,u}_{l,j}}{\bf{h}}^{\rm b,r,u\; H}_{l,j}}  + \frac{{\sigma _\eta ^2}}{M}{\bf{I}}} \right]}^{-1}  {\sum\limits_{k \in \zeta \left( l \right)} {{{\bf{h}}^{\rm b,r,u}_{l,k}}{\bf{h}}^{\rm b,r,u\; H}_{l,k}} }.
\end{equation}
To further consider the modulus-1 constraint of $\mathbf{e}_l$, we take the phase of the aforementioned eigenvector $\mathbf{u}$, i.e., $\mathbf{e}_l = \exp \{ j \angle(\mathbf{u}) \}$. However, the passive beamforming capability of the $l$th RIS alone is insufficient to eliminate the inter-UE interference among UEs indexed by $k \in \zeta(l)$. This issue will be addressed in the next subsection on DBF design.

\subsection{Interference Cancellation via DBF}
Although an identity matrix was initially chosen for DBF in the previous subsection, it needs to be redesigned as the MMSE precoding to address the remaining inter-UE and inter-RIS interference. In the above two parts, we have designed the ABF and passive beamforming at the BS and RISs, respectively. Given the antenna position, ABF matrix, and RIS phase shift matrices, the equivalent channel is given by
\begin{equation}
\mathbf{G}^{\text{H}}_{\text{equ}} = {\left( {{{\mathbf{H}}^{{\text{b,u}}}} + \sum\limits_{l = 1}^L {{\mathbf{H}}_l^{{\text{b,r}}}{{\mathbf{E}}_l}{\mathbf{H}}_l^{{\text{r,u}}}} } \right)^{\text{H}}}{\mathbf{V}},
\end{equation} 
where ${\bf{H}}_{}^{{\rm{b,u}}} = \left[{\bf h}_1^{{\rm{b,u}}} \cdots {\bf h}_K^{{\rm{b}},{\rm{u}}}\right]$. Note that we need to use the steering vectors of TFA subarrays in (\ref{TFASV}) rather than (\ref{FASV}) to reconstruct the channel here. Finally, the DBF matrix is 
\begin{equation}\label{TFADBF}
{\mathbf{W}} = {{\mathbf{G}}_{{\text{equ}}}}{\left( {{\mathbf{G}}_{{\text{equ}}}^{\text{H}}{{\mathbf{G}}_{{\text{equ}}}} + \frac{{\sigma _\eta ^2}}{P}{\mathbf{I}}} \right)^{ - 1}}.
\end{equation}
{\color{black}Similarly, DBF can also be optimized using the method presented in Sec. \ref{sec:opt}. For TFA systems, this offers two signal processing approaches: 1) adopting the closed-form solutions for all variables to reduce computational complexity, or 2) performing optimization between RIS and DBF alternately to achieve better performance.}
\subsection{\color{black}Overall Algorithm and Complexity of the TFAS}
{\color{black}
\subsubsection{Algorithm Summary}
For clarity, the signal processing flow chart of the RIS-assisted TFA system (TFAS) is given in Fig. \ref{FlowChart2}. Since no iteration is required and closed-form solutions exist, the computational complexity is significantly reduced compared to the FA optimization procedure in Fig. \ref{FlowChart1}. 
\subsubsection{Complexity Analysis}
The main computational complexity of the proposed RIS-assisted TFAS lies in the RIS passive design and matrix inversion in (\ref{TFADBF}). For Rayleigh-Ritz-theorem-based  RIS passive beamforming and MMES DBF, the complexity order is just $\mathcal{O}\left\{ {\max \left[ {{K^3},L{M^3},N} \right]} \right\}$. If we consider optimization for RIS and DBF, the corresponding complexity order for one iteration improves to $\mathcal{O}\left\{ {\max \left[ I_{\rm FP}{{K^{4.5}}\log\left(\varepsilon\right),I_{\rm FP}{{\left( {LM} \right)}^{4.5}}\log\left(1/\varepsilon\right),N} \right]} \right\}$.

\subsubsection{Optimality Analysis}
For single-RIS and single-user TFAS configuration, the GL-based antenna position and ABF design is optimal as shown in \cite{FAPortSelRIS}. While for the multiple RISs and UEs configuration, the optimality cannot be guaranteed due to the following reasons.
\begin{itemize}
\item We split the whole large antenna array to $K$ TFA subarrays with Sub-Con structure to serve $K$ UEs. This may cause some performance loss compared to the joint design of a large FA antenna array with a Full-Con structure and free antenna element movement.
\item Different RISs are assigned to different UEs, which may cause some performance loss compared to the case where each RIS serves all UEs.
\item We employ ABF for each TFA subarray independently along with a linear DBF, which may also suffer performance loss compared to the optimization-based beamforming design.
\item Only statistical CSI (LoS) is utilized in the TFAS design, which causes performance loss when the NLoS channel components are large.
\end{itemize}
}

\begin{figure}[t]
	\centering
	\includegraphics[scale=0.6]{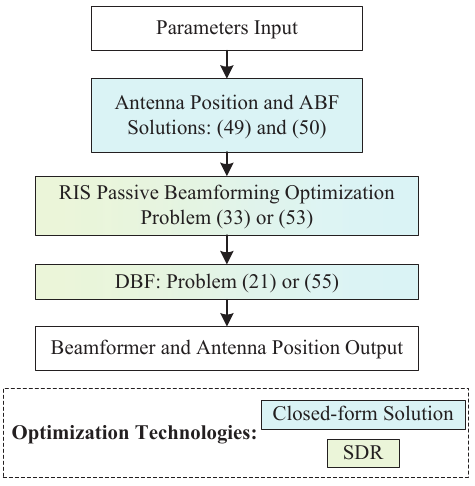}
	\caption{\color{black}Signal processing flow chart of the proposed RIS-assisted TFAS.}\label{FlowChart2}
	\vspace{-4mm}
\end{figure}

\vspace{-2mm}
\section{Simulation Results}\label{sec:results}
In this section, numerical results for both optimization-based FAS in Section \ref{sec:opt} and low-complexity GL-based TFAS in Section \ref{sec:tfa} are presented. {\color{black}For ease of comparison, double RISs are considered to serve three UEs in all schemes.} {\color{black}Part of the simulation parameters are given in Table \ref{SimSetup}, and other unspecified parameters will be noted in the figure captions.}
\renewcommand{\arraystretch}{1.2}
\begin{table}[htbp]
	\color{black}
	\caption{Simulation parameters setup.}
	\centering
	\begin{threeparttable}
		\begin{tabular}{*{3}{>{\centering\arraybackslash}m{2cm}>{\centering\arraybackslash}m{2cm}>{\centering\arraybackslash}m{3.5cm}}} 
			\Xhline{2pt}
			\textbf{Parameter} & \textbf{Value} & \textbf{Description} \\
			\Xhline{1pt}
			$f_c$ & 3.5 GHz & Carrier frequency \\
			$\lambda$  & 0.0857 m & Wavelength\\
			$N$ &  24  & Transmit antenna number  \\
			$K$ & 3 & RF chain and UE number \\
			$L$ & 2 & RIS number\\
			$M = M_1 \times M_2$ & 16 = 4$\times$4 & RIS unit number \\
			/ & ($0^{\rm o}$, $0^{\rm o}$, 0 m)  & Location of the transmit array\\
			\multirow{3}{*}{$(\theta^{\rm b,u}_k,\ \phi^{\rm b,u}_k,\ r^{\rm b,u}_k)$} 
& ($80^{\rm o}$, $0^{\rm o}$, 10 m)  & \multirow{3}{*}{Location of the UEs}  \\
& ($90^{\rm o}$, $0^{\rm o}$, 10 m)  &  \\
& ($100^{\rm o}$, $0^{\rm o}$, 10 m)  &  \\
			\multirow{2}{*}{$(\theta^{\rm b,r}_l,\ \phi^{\rm b,r}_l,\ r^{\rm b,r}_l)$} 
& ($10^{\rm o}$, $0^{\rm o}$, 5 m)  & \multirow{2}{*}{Location of the RISs}  \\
& ($170^{\rm o}$, $0^{\rm o}$, 5 m)  &  \\
$D = (N-1)\lambda$ & 1.9714 m & Maximum array aperture \\
$\delta = \lambda/2$ & 0.0429 m & Minimum element spacing\\
$\beta_0$ & 40 dB & Large-scale path loss \cite{3GPP}\\ 
$(\vartheta_{k}^{\text{b,u}},\, \vartheta_{l} ^{\text{b,r}}, \, \vartheta_{l,k} ^{\text{r,u}})$ & $(2.5,\,1.7,\,2.5)$ & Path loss exponents \cite{INCan}\\
$\sigma_\eta^2$ & -174 dBm/Hz & Noise power spectral density\\
$\rho$  & $10^{-4}$  & Convergence threshold \\
$I_{\rm MM}^{\rm max},\;I_{\rm BCD}^{\rm max}$ & 50, 20 & Maximum iteration number\\
			\Xhline{2pt}
		\end{tabular}%
	\end{threeparttable}
		\label{SimSetup}%
\end{table}

{\color{black}Next, we will present a comparison of this RIS-assisted MU-MISO system under different antenna types, hardware architectures, RIS number, and optimization methods. The meanings of the following scheme abbreviations are listed as follows
\begin{itemize}
\item Antenna Types: {\bf FA} refers to the fluid antenna proposed in Sec. \ref{sec:opt} of this paper, {\bf TFA} denotes the telescopic fluid antenna introduced in Sec. \ref{sec:tfa}, and {\bf FPA} represents the conventional fixed-position antenna.
\item Hardware Architecture: We use {\bf FD}, {\bf Full-Con}, and {\bf Sub-Con} to respectively represent the fully digital, fully-connected, sub-connected transmitter architecture for HBF.
\item RIS Number:  {\bf Single-RIS} denotes only one RIS is employed, i.e., $L=1$. This case can only be discussed in the FA antenna type. {\bf Double-RIS} represents two RISs are employed, i.e., $L = 2$. {\bf None RIS} indicates the case without RIS, i.e., $L=0$.
\item Signal Processing Methods: For FA and FPA antenna types, the beamformer is obtained through the optimization algorithm in Sec. \ref{sec:opt}. Therefore, we do not introduce special notation for the signal processing methods of FA and FPA systems. While for TFAS, the TFA position and active beamformer can be obtained through closed-form solution in (\ref{TFAPos}), (\ref{TFAABF}), and (\ref{TFADBF}), labeled as {\bf CFS}. While the RIS passive beamforming and DBF can also be optimized alternately by the optimization algorithm in Sec. \ref{sec:opt}, labeled as {\bf Opt}. Lastly, {\bf random phase} specifically refers to the RIS phase-shift coefficient matrix with randomly assigned values.

\end{itemize}

 }

\begin{figure}[ht]
	\centering
	\includegraphics[scale=0.8]{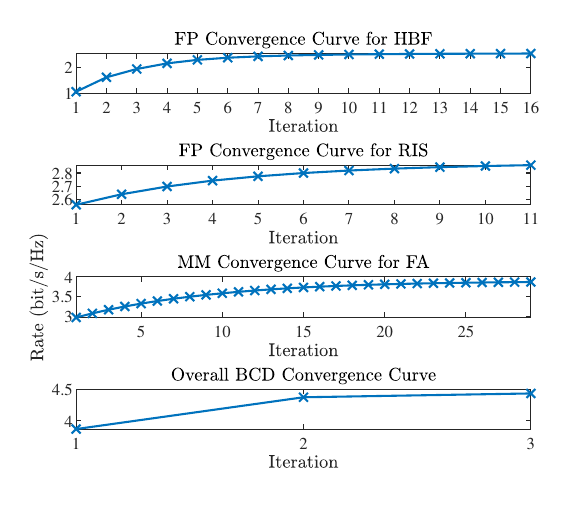}
	\caption{Convergence curves of the proposed scheme in Sec. \ref{sec:opt}.}\label{SimFig1}
	\vspace{-3mm}
\end{figure}

\begin{figure}[ht]
	\centering
	\includegraphics[scale=0.8]{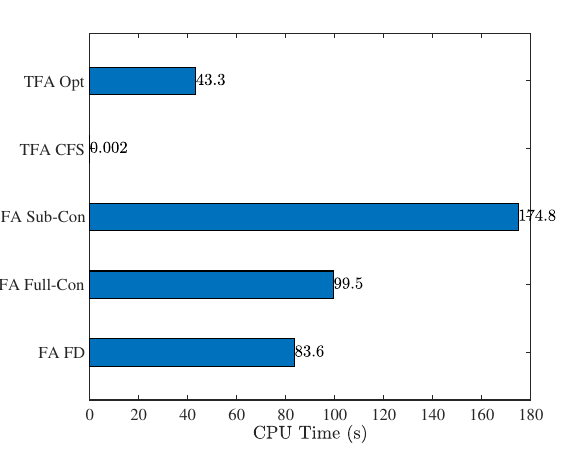}
	\caption{CPU time comparisons among different optimization schemes.}\label{CPUTim}
	\vspace{-3mm}
\end{figure}

\begin{figure}[ht]
	\centering
	\includegraphics[scale=0.8]{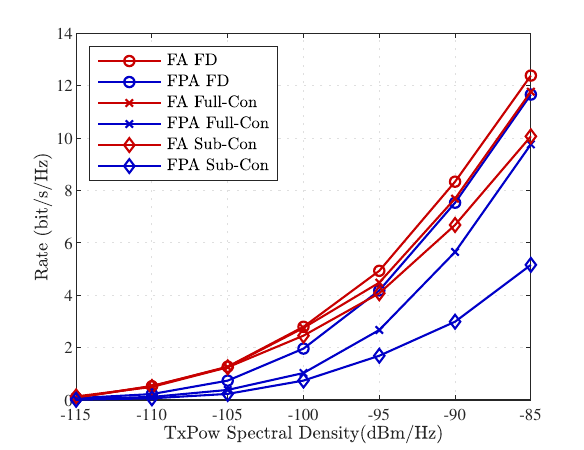}
	\caption{Sum rate under different transmitter architectures and antenna types.}\label{SimFig2}
	\vspace{-3mm}
\end{figure}

\begin{figure}[ht]
	\centering
	\includegraphics[scale=0.8]{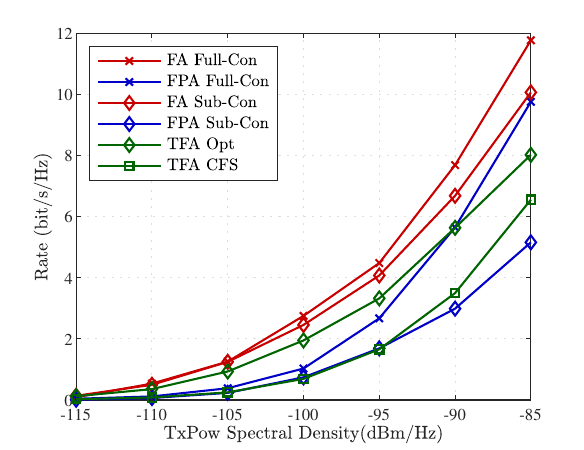}
	\caption{Sum rate between optimization-based FAS and GL-based TFAS.}\label{SimFig4}
	\vspace{-3mm}
\end{figure}

\begin{figure}[ht]
	\centering
	\includegraphics[scale=0.8]{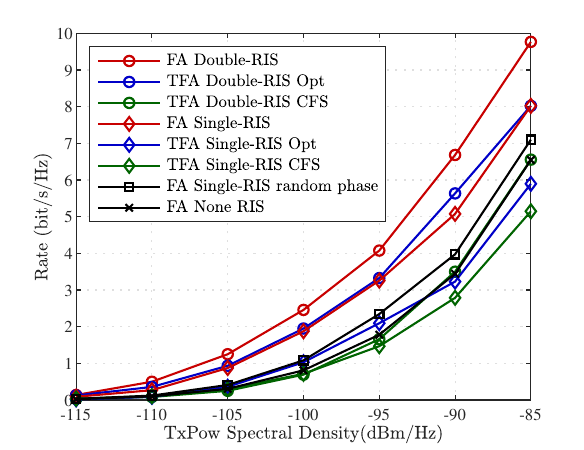}
	\caption{Sum rate under sub-connected architecture among different RIS configurations.}\label{SimFig3}
	\vspace{-3mm}
\end{figure}

\begin{figure}[ht]
	\centering
	\includegraphics[scale=0.8]{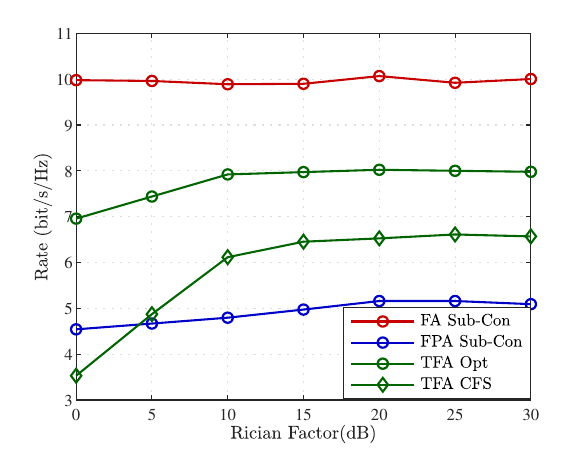}
	\caption{Sum rate v.s. Rician factor between FAS and TFAS.}\label{SimFig5}
	\vspace{-3mm}
\end{figure}

\begin{figure}[ht]
	\centering
	\includegraphics[scale=0.8]{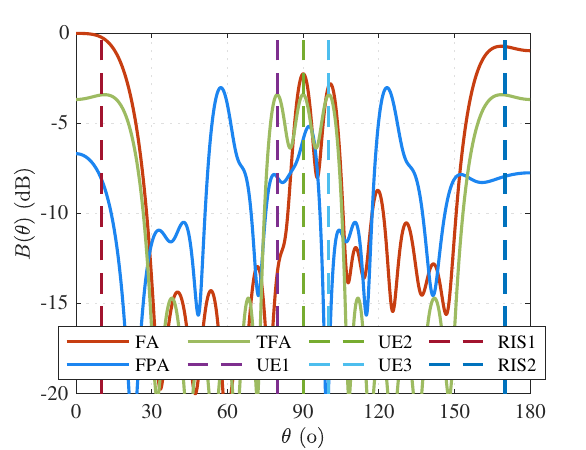}
	\caption{Normalized ABF beampattern between FAS, TFAS, and FPA.}\label{SimFig6}
	\vspace{-3mm}
\end{figure}

In Fig.~\ref{SimFig1}, we study the convergence behavior of the BCD algorithm and its three subproblems for the sub-connected HBF architecture when $P= -95~{\rm dBm}$ and $\kappa = 20~{\rm dB}$. Overall, the BCD algorithm requires only $2$ to $3$ iterations to converge, whereas the FA optimization needs around $20$ MM iterations. In contrast, active and passive beamforming both require nearly $10$ FP iterations for convergence. {\color{black}In Fig. \ref{CPUTim}, the average (central processing unit) CPU running time of different schemes in this paper are plotted for comparison in a laptop platform with 4.05 GHz Apple M3 Pro CPU and 18 GB RAM. First, it can be observed that the GL-based TFA schemes are of much less running time than the optimization-based FA schemes owing to the closed-form solutions of ABF and antenna positions. Moreover, the HBF  architecture, particularly in its sub-connected configuration, inherently requires longer running time than fully digital architecture due to its dual beamforming matrices and the additional iterative processes they necessitate.}

In Fig.~\ref{SimFig2}, results are provided for the sum rate of the proposed {\color{black}RIS-assisted FAS} under different transmitter architectures when $\kappa= 20~{\rm dB}$. First, as expected, the performance gains of different architectures follow the expected order as FD, Full-Con, and then Sub-Con. Also, FA outperforms its FPA counterpart across all three different architectures, with the largest gap in the sub-connected architecture. Additionally, the performance of the Sub-Con FAS can be comparable to that of the Full-Con FPA system, while the Full-Con FAS can achieve performance on par with the FD FPA system.

Fig.~\ref{SimFig4} compares the rate of the optimization-based FAS and the GL-based TFAS under the Sub-Con architecture with $\kappa= 20{\rm dB}$. TFAS has access only to the statistical CSI while FAS with full CSI provides the performance upper bound. A significant performance gap is observed between optimization-based FAS and Rayleigh-Ritz-based TFAS. This gap arises as the  Rayleigh-Ritz-based TFAS neither suppresses inter-UE and inter-RIS interference nor jointly optimizes the two RISs. However, the Rayleigh-Ritz-based TFAS provides gains over the optimization-based FPA system while maintaining extremely low computational complexity. Additionally, the performance of the Rayleigh-Ritz-based TFAS can be greatly enhanced to closely match that of the Full-Con FPA system through the joint optimization of the double-RIS employed in its FAS counterpart. But even with these performance measures, TFAS still falls short of its FAS counterpart due to the limitations in position switching and the absence of complete CSI.

The significance of RIS is studied by the results in Fig.~\ref{SimFig3} with the sub-connected architecture when $\kappa= 20~{\rm dB}$. With the same optimization method in Sec. \ref{sec:opt}, the double-RIS-assisted FAS offers better rate performance compared to its single-RIS and no-RIS counterparts. From the perspective of RIS phase shifts, although random RIS phase shifts are less effective than the optimized ones, they still provide additional gain compared to the system without RIS. {\color{black}For the double-RIS TFAS system, we have included performance comparisons with its single-RIS counterpart in \cite{FAPortSelRIS}. While the single-RIS TFAS was originally designed for single-UE scenarios, our results demonstrate that it can also be performed in MU scenarios but with performance degradation similar to FAS. Moreover, we note that while the achievable rate metric effectively captures communication performance for the intended UEs, it does not fully reflect the spatial leakage interference effects, as clearly illustrated in Fig. \ref{SpatialInt}.}

The rate comparisons versus $\kappa$ {\color{black}with transmit power spectral density $P = -85 $ dBm/Hz} are given in Fig.~\ref{SimFig5}. It can be seen that the performances of FAS, TFAS and FPA based on optimization methods show less fluctuation due to utilization of the complete CSI. {\color{black}In contrast, TFAS with closed-form solutions, which only utilizes the LoS part of the CSI, experience declines in performance as $\kappa$ decreases.} Specifically, the performance of the former falls below that of  FPA when $\kappa$ drops to approximately $4~{\rm dB}$ or lower.

Finally, in Fig.~\ref{SimFig6}, we compare the beampattern of the optimization-based FAS, GL-based TFAS, and FPA when $P= -85{\rm dBm/Hz}$, $\kappa= 30{\rm dB}$. It can be observed that despite complex optimizations, FPA fails to project precise beams towards the directions of the UEs and two RISs, showing significant beam misalignment. By contrast, TFAS generates five beams with identical full array gain towards all terminals through delicate ML and GL design. In addition, optimization-based FAS can provide extra power allocation among these terminals, where UE 1 is allocated less power. %This joint design is difficult to implement in TFAS, as TFAS designs sub-arrays, RISs, and UE groups independently.

\section{Conclusion}\label{sec:conclude}
In this paper, we investigated {\color{black} a RIS-assisted MU-MISO FAS.} First of all, a sum-rate maximization problem was formulated w.r.t.~HBF, RIS phase shift, and antenna position designs. This nonconvex problem was solved by adopting the BCD and FP frameworks combined with SDR and MM relaxation algorithms. To reduce movement and computational complexity, a simplified TFA was proposed and optimized in an extremely low-complexity manner based on the GL effect. Simulation results demonstrated the performance gains of both FAS and TFAS over the conventional FPA system. {\color{black}The proposed FAS and TFAS schemes both demonstrate significant performance improvements for multi-directional beamforming systems. Beyond the current applications, these approaches show promising potential for extension to other scenarios requiring multi-directional beamforming capabilities, such as multi-access systems, ISAC, and UAV networks. Future work may further bridge the FA-TFA performance gap through advanced resource allocation and array optimization techniques}

\appendices
\section{Proof of Proposition \ref{prop1}}\label{AppdixA}
Considering the original optimization objective function without the quadratic transform, the Lagrangian function of Problem (\ref{opt2}) can be formulated as
\begin{multline}
	\mathcal{L}\left( {{{\mathbf{\Gamma }}_{{\text{all}}}},\mu,{\mathbf{L}_\text{all}}} \right) = - \sum\limits_{k = 1}^K {{{\log }_2}\left( {1 + \frac{{{\text{Tr}}\left( {{\mathbf{R}}_k^g{{\mathbf{\Gamma }}_k}} \right)}}{\sum\limits_{j \ne k} {{\text{Tr}}\left( {{\mathbf{R}}_k^g{{\mathbf{\Gamma }}_j}} \right) + \sigma _\eta ^2} }} \right)}\\
	+ \mu\left[ {{\text{Tr}}\left( {\sum\limits_{k = 1}^K {{{\mathbf{\Gamma }}_k}} } \right)} - P \right] - \text{Tr}\left(\sum\limits_{k = 1}^K {{{\mathbf{L}}_k}{{\mathbf{\Gamma }}_k}}\right),
\end{multline}
where $\mu \geq 0$ is the Lagrange multiplier for the constraint (\ref{opt2}b), ${{\mathbf{L}}_{{\text{all}}}} = \left[{\mathbf{L}}_1 \cdots {\mathbf{L}}_K\right]$ with each ${{{\mathbf{L}}_k}} \in {\mathbb{H}^{N \times N}} \succeq \mathbf{0}$ denoting the Lagrange multiplier matrix for the positive semi-definite constraint (\ref{opt2}c). Then part of the Karush-Kuhn-Tucker (KKT) conditions can be expressed as 
\begin{subnumcases} {}
{{\mathbf{L}}_k^{\star}} = u^{\star}{\mathbf{I}} + {{\mathbf{T}}_1} - {{\mathbf{T}}_2},\\
%&{\text{Tr}}\left( {\sum\limits_{k = 1}^K {{{\mathbf{\Gamma }}_k^{\star}}} } \right) \leqslant P,\\
{{\mathbf{\Gamma }}_k^{\star}} \succeq  \mathbf{0}, \mathbf{L}_k^{\star} \succeq \mathbf{0}, \; \mu^{\star} \geq 0, \\
% \mu^{\star} \geq 0, \\
%&\mathbf{L}_k^{\star} \succeq \mathbf{0}, \\
%&\mu^{\star} \left[{\text{Tr}}\left( {\sum\limits_{k = 1}^K {{{\mathbf{\Gamma }}_k^{\star}}} } \right) - P\right] = 0, \\
\text{Tr} \left(\mathbf{\Gamma}_k^{\star} \mathbf{L}_k^{\star}\right) = 0,
\label{kkt}%
\end{subnumcases}
where $\mathbf{T}_1$ and $\mathbf{T}_2$ can be expressed as
\begin{align}
{{\mathbf{T}}_1} &= \frac{1}{{{\ln 2}}}\sum\limits_{j \ne k} {\frac{1}{{\sum\limits_{i = 1}^K {{\text{Tr}}\left( {{\mathbf{R}}_j^g{{\mathbf{\Gamma }}_i^{\star}}} \right) + \sigma _\eta ^2} }}}  \times \frac{{{\text{Tr}}\left( {{\mathbf{R}}_j^g{{\mathbf{\Gamma }}_j^{\star}}} \right){\mathbf{R}}_j^g}}{{\sum\limits_{i \ne j}^K {{\text{Tr}}\left( {{\mathbf{R}}_j^g{{\mathbf{\Gamma }}_i^{\star}}} \right) + \sigma _\eta ^2} }},\\
{{\mathbf{T}}_2} &= \frac{1}{{{\ln 2}}}\frac{{{\mathbf{R}}_k^g}}{{\sum\limits_{i = 1}^K {{\text{Tr}}\left( {{\mathbf{R}}_i^g{{\mathbf{\Gamma }}_i^{\star}}} \right)}  + \sigma _\eta ^2}}.
\end{align}
It is evident that $\mathbf{T}_2$ is rank-1. 
Additionally, since $\mathbf{\Gamma}_k^{\star}$ and $\mathbf{L}_k^{\star}$ are positive semi-definite Hermitian matrices, (\ref{kkt}) can be rewritten as $\mathbf{\Gamma}_k^{\star} \mathbf{L}_k^{\star} = \mathbf{0}$, indicating that 
\begin{equation}\label{dimlaw}
{\text{rank}}\left( {{\mathbf{L}}_k^ \star } \right) \leq \dim \left( {{\text{Null}}\left( {{\mathbf{\Gamma }}_k^ \star } \right)} \right) = N - {\text{rank}}\left( {{\mathbf{\Gamma }}_k^ \star } \right),
\end{equation}
where $\dim\left(\cdot\right)$ represents the operation of taking space dimension, and ${\text{Null}}\left(\cdot\right)$ stands for the null-space operation. Let $\mathbf{N}={\text{Null}} \left(u\mathbf{I} + \mathbf{T}_1\right) = \left[{\mathbf{n}}_1 \cdots {\mathbf{n}}_{N - r}\right]$, where $r = {\text{rank}}\left(u\mathbf{I} + \mathbf{T}_1\right)$. Next, we will discuss the value of $r$. We first consider $r = N$.

In this case, $u\mathbf{I} + \mathbf{T}_1$ is full-rank. Based on the aforementioned analysis, the rank-1 property of ${{\mathbf{\Gamma }}_k^ \star }$ can be proved by 
\begin{multline}\label{rankineq}
{\text{rank}}\left( {{\mathbf{L}}_k^ \star } \right) + {\text{rank}}\left( {{{\mathbf{T}}_2}} \right)\\
	\geqslant {\text{rank}}\left( {{\mathbf{L}}_k^ \star  + {{\mathbf{T}}_2}} \right) = {\text{rank}}\left( {u{\mathbf{I}} + {{\mathbf{T}}_1}} \right) = N.
\end{multline}
Combining (\ref{rankineq}), (\ref{dimlaw}), and $\text{rank}\left(\mathbf{T}_2\right) = 1$, gives ${\text{rank}}\left( {{\mathbf{\Gamma }}_k^ \star } \right) \leq 1$. Since $ {{\mathbf{\Gamma }}_k^ \star }  \ne \mathbf{0}$, we have ${\text{rank}}\left( {{\mathbf{\Gamma }}_k^ \star } \right) \geq 1$, thereby obtaining the conclusion that ${\text{rank}}\left( {{\mathbf{\Gamma }}_k^ \star } \right) = 1$.

Now, we consider the case $r < N$. In this case, the null-space of $u\mathbf{I} + \mathbf{T}_1$ is not empty. Since $\mathbf{L}_k^* \succeq  \mathbf{0}$, we arrive at
\begin{equation}
\mathbf{n}_n^{\text H} \mathbf{L}_k^* \mathbf{n}_n = \mathbf{n}_n^{\text H} \left(u^{\star}{\mathbf{I}} + {{\mathbf{T}}_1} - {{\mathbf{T}}_2}\right) \mathbf{n}_n \geq 0. 
\end{equation}
Therefore, we have 
\begin{equation}
\mathbf{n}_n^{\text H} {{\mathbf{T}}_2} \mathbf{n}_n = 0  \Rightarrow \mathbf{T}_2 \mathbf{N} = \mathbf{0}.
\end{equation}
Furthermore, according to the rank inequality, we have
\begin{equation}
r-1 \leq	\text{rank}\left(\mathbf{L}_k^*\right) \leq r+1.
\end{equation}
Moreover, letting $\mathbf{M}$ denote the null-space of $\mathbf{L}_k^*$, the dimension of $\mathbf{M}$ follows that
\begin{equation}
\dim\left(\mathbf{M}\right) = N - \text{rank}\left(\mathbf{L}_k^*\right) \mathop  \geq \limits^{(a)} \dim \left(\mathbf{N}\right) = N -r,
\end{equation}
where inequality (a) holds since the null-space of $\mathbf{L}_k^*$ is larger than that of $u^* \mathbf{I} + \mathbf{T}_1$. Thus, the upper bound of $\text{rank}\left(\mathbf{L}_k^*\right)$ can be further tightened to $r$. Assuming $\text{rank}\left(\mathbf{L}_k^*\right) = r$, we can deduce that $\mathbf{N} = \mathbf{M}$. Otherwise, if $\text{rank}\left(\mathbf{L}_k^*\right) = r - 1$, null-space $\mathbf{M}$ can be composed of $\mathbf{N}$ and $\bm{\tau}$, where $\bm{\tau}$ is a unit-form vector that lies in the null-space of $\left(\mathbf{L}_k^*\right)$.

According to \cite{Rank1Proof}, we can always find an optimal solution
\begin{equation}
{\mathbf{\Gamma }}_k^ \star  = \sum\limits_{n = 1}^{N - r} {{a_n}{{\mathbf{n}}_n}{\mathbf{n}}_n^{\text{H}}}  + b \bm{\tau} \bm{\tau}^{\textbf{H}},
\end{equation}
where $a_n \geq 0, \forall n \in \left[1,N-r\right]$, $b \geq 0$. When $\text{rank}\left(\mathbf{L}_k^*\right) = r$, $b = 0$, the rank-1 solution can be constructed by setting only one $a_n >0$ while others equal $0$. According to \cite{Rank1Proof}, no actual information is transmitted to the UEs, and the achievable rate is zero. Hence, the choice of specific $a_n$ makes no difference. On the other hand, when $\text{rank}\left(\mathbf{L}_k^*\right) = r - 1$, ${\mathbf{\Gamma }}_k^ \star  =  b \bm{\tau} \bm{\tau}^{\textbf{H}}$ is also an optimal solution with rank-1 property.

\section{Auxiliary Variables and Functions in ${{\overset{\lower0.5em\hbox{$\smash{\scriptscriptstyle\smile}$}}{A} }_k}\left( {{\mathbf{z}},{{\mathbf{f}}_k}} \right)$}\label{AppdixB}

To start with, the constant term ${a_0}\left( {{\mathbf{f}}_k^{}}\right)$ is expressed as
\begin{equation}
{a_0}\left( {{\mathbf{f}}_k^{}} \right) = {\left| {{\mathbf{f}}_k^{\text{H}}\left( {\tilde \chi _k^{{\text{b,u}}}{\mathbf{\tilde h}}_k^{{\text{b}},{\text{u}}} + \sum\limits_l {\tilde \chi _l^{{\text{b,r}}}{\mathbf{\tilde H}}_l^{{\text{b,r}}}{{\mathbf{E}}_l}{\mathbf{h}}_{l,k}^{{\text{r}},{\text{u}}}} } \right)} \right|^2}.
\end{equation}
Also, trigonometric functions w.r.t.~$\mathbf{z}$ are given by
\begin{equation}
{g_1}\left( {{\theta _1},{\theta _2},{\mathbf{z}},{{\mathbf{f}}_k}} \right) = g_1^{\cos }\left( {{\theta _1},{\theta _2},{\mathbf{z}},{{\mathbf{f}}_k}} \right) + \jmath g_1^{\sin }\left( {{\theta _1},{\theta _2},{\mathbf{z}},{{\mathbf{f}}_k}} \right),
\end{equation}
\begin{multline}
{g_2}\left( {\theta ,{\mathbf{z}},{\mathbf{b}}\left( {{{\mathbf{f}}_k}} \right)} \right) = 2\sum\limits_{n = 1}^N {\Re \left\{ \left[\mathbf{b}\right]_n \right\}\cos \left( {{z_n}\hat \theta } \right)}\\
+ \Im \left\{ \left[\mathbf{b}\right]_n \right\}\sin \left( {{z_n}\hat \theta } \right),
\end{multline}
where 
\begin{multline}
{g_1^{\cos / \sin}}\left( {{\theta _1},{\theta _2},{\mathbf{z}},{{\mathbf{f}}_k}} \right) = \sum\limits_{n = 1}^N {\sum\limits_{m = 1}^N {\left| \left[ \mathbf{f}_k \right]_n \left[ \mathbf{f}_k \right]_m \right| \times } }\\
 {\cos / \sin \left( {{u_{k,{\theta _1},{\theta _2}}}\left( {{z_n},{z_m}} \right)} \right),}
\end{multline}
with ${u_{k,{\theta _1},{\theta _2}}}\left( {{z_n},{z_m}} \right) = {{\hat \theta }_2}{z_n} - {{\hat \theta }_1}{z_m} - \left( {\angle {\left[\mathbf{f}_k \right]_n }  - \angle {\left[ \mathbf{f}_k\right]_m}  } \right)$.

Some auxiliary variables in (\ref{Az}) are listed below:
\begin{align}
{\mathbf{b}}_k^{\text{H}}\left( {{{\mathbf{f}}_k}} \right) &= \bar \chi _k^{{\text{b,u}}}{\left( {\tilde \chi _k^{{\text{b,u}}}{\mathbf{\tilde h}}_k^{{\text{b,u}}} + \sum\limits_l {\tilde \chi _l^{{\text{b,r}}}{\mathbf{\tilde H}}_l^{{\text{b,r}}}{{\mathbf{E}}_l}{\mathbf{h}}_{l,k}^{{\text{r,u}}}} } \right)^{\text{H}}}{{\mathbf{f}}_k}{\mathbf{f}}_k^{\text{H}} ,\\
{\mathbf{c}}_{k,l}^{\text{H}}\left( {{{\mathbf{f}}_k}} \right) &= \bar \chi _l^{{\text{b,r}}}{{\mathbf{a}}_M}{(\theta _l^{{\text{b}},{\text{r}}},\phi _l^{{\text{b}},{\text{r}}})^{\text{H}}}{{\mathbf{E}}_l}{\mathbf{h}}_{l,k}^{{\text{r}},{\text{u}}} \times\notag\\
&\quad\quad{\left( {\tilde \chi _k^{{\text{b,u}}}{\mathbf{\tilde h}}_k^{{\text{b,u}}} + \sum\limits_{{l_1} = 1}^L {\tilde \chi _{{l_1}}^{{\text{b,r}}}{\mathbf{\tilde H}}_{{l_1}}^{{\text{b,r}}}{{\mathbf{E}}_{{l_1}}}{\mathbf{h}}_{{l_1},k}^{{\text{r,u}}}} } \right)^{\text{H}}}{{\mathbf{f}}_k}{\mathbf{f}}_k^{\text{H}},\\
\hat{\theta} & =  - 2\pi \cos \theta /\lambda,~{a_{l}}  ={{\mathbf{h}}{_{l,k}^{{\text{r}},{\text{u}}}}^{\text{H}}}{\mathbf{E}}_l^{\text{H}}{{\mathbf{a}}_M}(\theta _l^{{\text{b}},{\text{r}}},\phi _l^{{\text{b}},{\text{r}}}).
\end{align}
The coefficient matrices in (\ref{Az2}) are given in (\ref{Rz})--(\ref{cz}), where ${A_l} = \bar \chi _k^{{\text{b}},{\text{u}}}\bar \chi _l^{{\text{b}},{\text{r}}\;{\text{*}}}{a_l}$, and $B_{{l_1},{l_2}}^{} = \bar \chi _{{l_1}}^{{\text{b}},{\text{r}}}\bar \chi _{{l_2}}^{{\text{b}},{\text{r*}}}a_{{l_1}}^*{a_{{l_2}}}$.
\begin{figure*}[htbp]
\begin{align}
	{\mathbf{\bar R}}_k^z =& {\left| {\bar \chi _k^{{\text{b}},{\text{u}}}} \right|^2}{{\mathbf{R}}_{g1}}\left( {\theta _k^{{\text{b}},{\text{u}}},\theta _k^{{\text{b}},{\text{u}}},{\mathbf{f}}_k^{}} \right) + 2\sum\limits_l {\left[ {\Re \left\{ {{A_l}} \right\} + \Im \left\{ {{A_l}} \right\}} \right]{{\mathbf{R}}_{g1}}\left( {\theta _l^{{\text{b}},{\text{r}}},\theta _k^{{\text{b}},{\text{u}}},{\mathbf{f}}_k^{}} \right)}  + \sum\limits_l {{{\left| {\bar \chi _l^{{\text{b}},{\text{r}}}{a_l}} \right|}^2}{{\mathbf{R}}_{g1}}\left( {\theta _l^{{\text{b}},{\text{r}}},\theta _l^{{\text{b}},{\text{r}}},{\mathbf{f}}_k^{}} \right)}    \notag\\
	&+ {{\mathbf{R}}_{g2}}\left( {\theta _k^{{\text{b}},{\text{u}}},{\mathbf{b}}_k^{\text{H}}\left( {{\mathbf{f}}_k^{}} \right)} \right) + \sum\limits_l {{{\mathbf{R}}_{g2}}\left( {\theta _l^{{\text{b}},{\text{r}}},{\mathbf{c}}_{k,l}^{\text{H}}\left( {{\mathbf{f}}_k^{}} \right)} \right)}  +  + 2\sum\limits_{{l_1}} {\sum\limits_{{l_2} > {l_1}} {\left[ {\Re \left\{ {{B_l}} \right\} + \Im \left\{ {{B_l}} \right\}} \right]{{\mathbf{R}}_{g1}}\left( {\theta _{{l_2}}^{{\text{b}},{\text{r}}},\theta _{{l_1}}^{{\text{b}},{\text{r}}},{\mathbf{f}}_k^{}} \right)} }, 
\label{Rz} \\
	{\bf{r}}_k^z =& {\left| {\bar \chi _k^{{\rm{b,u}}}} \right|^2}{\bf{r}}_{k,g1}^{z,\cos }\left( {\theta _k^{{\rm{b}},{\rm{u}}},\theta _k^{{\rm{b}},{\rm{u}}},{{\bf{z}}^i},{\bf{f}}_k^{}} \right) + 
	2\sum\limits_l {\Re \left\{A_l\right\}} {\bf{r}}_{k,g1}^{z,\cos }\left( {\theta _l^{{\rm{b}},{\rm{r}}},\theta _k^{{\rm{b}},{\rm{u}}},{{\bf{z}}^i},{\bf{f}}_k^{}} \right) + 
	2\sum\limits_l {\Im \left\{A_l\right\}} { {\bf{r}}_{k,g1}^{z,\sin }\left( {\theta _l^{{\rm{b}},{\rm{r}}},\theta _k^{{\rm{b}},{\rm{u}}},{{\bf{z}}^i},{\bf{f}}_k^{}} \right)}  \notag \\
	&+ 2\sum\limits_{{l_1}} {\sum\limits_{{l_2} > {l_1}} {\Re \left\{ B_{l_1,l_2} \right\}{\bf{r}}_{k,g1}^{z,\cos }\left( {\theta _{{l_2}}^{{\rm{b}},{\rm{r}}},\theta _{{l_1}}^{{\rm{b}},{\rm{r}}},{{\bf{z}}^i},{\bf{f}}_k^{}} \right) + } } \Im \left\{  B_{l_1,l_2} \right\}{\bf{r}}_{k,g1}^{z,\sin }\left( {\theta _{{l_2}}^{{\rm{b}},{\rm{r}}},\theta _{{l_1}}^{{\rm{b}},{\rm{r}}},{{\bf{z}}^i},{\bf{f}}_k^{}} \right)
	\notag \\
	&+ \sum\limits_l {{{\left| {\bar \chi _l^{{\rm{b,r}}}{a_l}} \right|}^2}{{\bf{r}}^{z,\cos}_{k,g1}}\left( {\theta _l^{{\rm{b}},{\rm{r}}},\theta _l^{{\rm{b,r}}},{{\bf{z}}^i},{\bf{f}}_k^{}} \right)} +
	{{\bf{r}}_{k,g2}^{z}}\left( {\theta _k^{{\rm{b}},{\rm{u}}},{{\bf{z}}^i},{\bf{b}}_k^{\rm{H}}\left( {{\bf{f}}_k^{}} \right)} \right) + 
	\sum\limits_l {{{\bf{r}}_{k,g2}^{z}}\left( {\theta _l^{{\rm{b}},{\rm{r}}},{{\bf{z}}^i},{\bf{c}}_{k,l}^{\rm{H}}\left( {{\bf{f}}_k^{}} \right)} \right)},  \label{rz} \\
	c_k^z = &
	{\left| {\bar \chi _k^{{\rm{b,u}}}} \right|^2}c_{k,g1}^{z,\cos }\left( {\theta _k^{{\rm{b}},{\rm{u}}},\theta _k^{{\rm{b}},{\rm{u}}},{{\bf{z}}^i},{\bf{f}}_k^{}} \right) 
	+ 2\sum\limits_l {\Re \left\{ A_l \right\}} c_{k,{g_1}}^{z,\cos }\left( {\theta _l^{{\rm{b}},{\rm{r}}},\theta _k^{{\rm{b}},{\rm{u}}},{{\bf{z}}^i},{\bf{f}}_k^{}} \right) 
	+ 2\sum\limits_l {\Im \left\{ A_l \right\}c_{k,{g_1}}^{z,\sin }\left( {\theta _l^{{\rm{b}},{\rm{r}}},\theta _k^{{\rm{b}},{\rm{u}}},{{\bf{z}}^i},{\bf{f}}_k^{}} \right)} \notag \\
	&+2\sum\limits_{{l_1}} {\sum\limits_{{l_2} > {l_1}} {\Re \left\{ B_{l_1,l_2} \right\}c_{k,{g_1}}^{z,\cos }\left( {\theta _{{l_2}}^{{\rm{b}},{\rm{r}}},\theta _{{l_1}}^{{\rm{b}},{\rm{r}}},{{\bf{z}}^i},{\bf{f}}_k^{}} \right) + } } \Im \left\{ B_{l_1,l_2} \right\}c_{k,{g_1}}^{z,\sin }\left( {\theta _{{l_2}}^{{\rm{b}},{\rm{r}}},\theta _{{l_1}}^{{\rm{b}},{\rm{r}}},{{\bf{z}}^i},{\bf{f}}_k^{}} \right) \notag \\
	&+\sum\limits_l {{{\left| {\bar \chi _l^{{\rm{b,r}}}{a_l}} \right|}^2}c_{k,{g_1}}^{z,\cos }\left( {\theta _l^{{\rm{b}},{\rm{r}}},\theta _l^{{\rm{b,r}}},{{\bf{z}}^i},{\bf{f}}_k^{}} \right)}  +
	c_{k,g2}^z\left( {\theta _k^{{\rm{b}},{\rm{u}}},{{\bf{z}}^i},{{\bf{b}}_k}\left( {{\bf{f}}_k^{}} \right)} \right) + \sum\limits_l {c_{k,g2}^z\left( {\theta _l^{{\rm{b}},{\rm{r}}},{{\bf{z}}^i},{{\bf{c}}_{k,l}}\left( {{\bf{f}}_k^{}} \right)} \right)} + {a_0}\left( {{{\bf{f}}_k}} \right)\label{cz}
\end{align}
\hrulefill
\end{figure*}
The auxiliary matrices in (\ref{Rz}) are expressed as
\begin{multline}
{{\mathbf{R}}_{g1}}\left( {{\theta _1},{\theta _2},{\mathbf{f}}} \right) =   2\hat \theta _1^{}\hat \theta _2^{}\left| {\mathbf{f}} \right|{\left| {\mathbf{f}} \right|^{\text{T}}} -  	\left( {\hat \theta _1^2 + \hat \theta _2^2} \right) \times\\
\sum\limits_{m = 1}^N {\left|  \left[ \mathbf{f} \right]_m \right|} {\text{diag}}\left\{ {\left| \left[ \mathbf{f} \right]_1 \right|,\left| \left[ \mathbf{f} \right]_2 \right|,\dots,\left| \left[ \mathbf{f} \right]_N \right|} \right\},
\end{multline}
\begin{equation}
{{\mathbf{R}}_{g2}}\left( {\theta ,{\mathbf{b}}} \right) =  2{\text{diag}}\left\{ {\left[ { - \Re \left\{ \left[\mathbf{b}\right]_n  \right\} + {\Im} \left\{ \left[\mathbf{b}\right]_n \right\}} \right]{{\hat \theta }^2}} \right\}.
\end{equation}
Also, the auxiliary vectors in (\ref{rz}) are formulated as
\begin{multline}
\left[{\bf{r}}_{k,g1}^{z,\cos/ \sin }\left( {{\theta _1},{\theta _2},{{\bf{z}}^i},{\bf{f}}} \right)\right]_n = \sum\limits_{m = 1}^N {\left| \left[ \mathbf{f} \right]_n \left[ \mathbf{f} \right]_m \right| \times } \\
\left\{ {\left[ { - \cos / \sin \left( {u\left( {z_n^i,z_m^i} \right)} \right) + \left( {{{\hat \theta }_2}z_n^i - {{\hat \theta }_1}z_m^i} \right)} \right]{{\hat \theta }_2} - } \right.\\
\left. {\left[ { - \cos / \sin \left( {u\left( {z_m^i,z_n^i} \right)} \right) + \left( {{{\hat \theta }_2}z_m^i - {{\hat \theta }_1}z_n^i} \right)} \right]{{\hat \theta }_1}} \right\},
\end{multline}
\begin{multline}
\left[{\mathbf{r}}_{k,g2}^z\left( {\theta ,{\mathbf{z}^{i}},{\mathbf{b}}} \right)\right]_n = 2\Re \left\{ \left[\mathbf{b}\right]_n \right\}\left[ { - \sin \left( {z_n^i\hat \theta } \right)\hat \theta  + z_n^i{{\hat \theta }^2}} \right]\\
- 2 \Im \left\{ \left[\mathbf{b}\right]_n \right\}\left[ - {\cos \left( {z_n^i\hat \theta } \right)\hat \theta  + z_n^i{{\hat \theta }^2}} \right].
\end{multline}
Finally, the auxiliary variables in (\ref{cz}) are denoted as
\begin{multline}
c_{k,g1}^{z,\cos/\sin}\left( {{\theta _1},{\theta _2},{\mathbf{z}^i},{\mathbf{f}}} \right) =  -	\frac{1}{2}\sum\limits_{n = 1}^N {\sum\limits_{m = 1}^N {\left| \left[ \mathbf{f} \right]_n \left[ \mathbf{f} \right]_m \right|} } {\left( {{{\hat \theta }_2}{z_n^i} - {{\hat \theta }_1}{z_m^i}} \right)^2}\\
+ \sum\limits_{n = 1}^N {\sum\limits_{m = 1}^N {\left| \left[ \mathbf{f} \right]_n \left[ \mathbf{f} \right]_m \right|\sin / \cos \left( {u\left( {z_n^i,z_m^i} \right)} \right)\left( {{{\hat \theta }_2}{z_n^i} - {{\hat \theta }_1}{z_m^i}} \right)} }\\
+ \sum\limits_{n = 1}^N {\sum\limits_{m = 1}^N {\left| \left[ \mathbf{f} \right]_n \left[ \mathbf{f} \right]_m \right|\cos / - \sin \left( {u\left( {z_n^i,z_m^i} \right)} \right)  } }, 
\end{multline}
\begin{multline}
c_{k,g2}^z\left( {\theta ,{\mathbf{z}^{i}},{\mathbf{b}}} \right) =\\
2\sum\limits_{n = 1}^N {\Re \left\{  \left[\mathbf{b}\right]_n \right\}\left[ {\cos \left( {z_n^i\hat \theta } \right) + \sin \left( {z_n^i\hat \theta } \right)z_n^i\hat \theta  - \frac{1}{2}{{\left( {z_n^i\hat \theta } \right)}^2}} \right]}\\
- \Im \left\{ \left[\mathbf{b}\right]_n \right\}\left[ { - \sin \left( {z_n^i\hat \theta } \right) + \cos \left( {z_n^i\hat \theta } \right)z_n^i\hat \theta  - \frac{1}{2}{{\left( {z_n^i\hat \theta } \right)}^2}} \right].
\end{multline}

\bibliographystyle{IEEEtran}
%\bibliography{IEEEabrv,Ref}


\end{document}